\documentclass[journal,twoside]{IEEEtran}
\usepackage{url}
\usepackage{graphicx}
\usepackage{amssymb}
\usepackage{amsmath}
\usepackage{amsfonts}
\usepackage{enumerate}
\usepackage{amsthm}
\usepackage{color}
\providecommand{\U}[1]{\protect\rule{.1in}{.1in}}

\theoremstyle{plain}
\newtheorem{theorem}{Theorem}
\newtheorem{corollary}[theorem]{Corollary}
\newtheorem{example}[theorem]{Example}

\newtheorem{proposition}[theorem]{Proposition}
\newtheorem{lemma}[theorem]{Lemma}
\newtheorem{definition}[theorem]{Definition}

\begin{document}
\markboth{IEEE Transactions on Information Theory,~Vol.~XX, No.~X, October~2013}{Lim and Comon: Blind Multilinear Identification}

\title{Blind Multilinear Identification}
\author{Lek-Heng~Lim$^{\ast}$ and ~Pierre~Comon$^{\ddagger}$%
,~\IEEEmembership{Fellow,~IEEE}\thanks{$\ast$
Lek-Heng Lim is with the Computational and Applied Mathematics Initiative,
Department of Statistics, University of Chicago, 5734 South University Avenue,
Chicago, IL 60637, USA. $\ddagger$  \textit{Corresponding Author}: Pierre Comon is with GIPSA-Lab CNRS
UMR5216, Grenoble Campus, BP.46, F-38402 St Martin d'Heres cedex, France.{}\par 
\hfill \copyright 2013  IEEE}}
\maketitle

\begin{abstract}
We discuss a technique that allows blind recovery of signals or blind
identification of mixtures in instances where such recovery or identification
were previously thought to be impossible: (i) closely located or highly
correlated sources in antenna array processing, (ii) highly correlated
spreading codes in CDMA radio communication, (iii) nearly dependent spectra in
fluorescent spectroscopy. This has important implications --- in the case of
antenna array processing, it allows for joint localization and extraction of
multiple sources from the measurement of a noisy mixture recorded on multiple
sensors in an entirely deterministic manner. In the case of CDMA, it allows
the possibility of having a number of users larger than the spreading gain. In
the case of fluorescent spectroscopy, it allows for detection of nearly
identical chemical constituents. The proposed technique involves the solution
of a bounded coherence low-rank multilinear approximation problem. We show that
bounded coherence allows us to establish existence and uniqueness of the
recovered solution. We will provide some statistical motivation for the
approximation problem and discuss greedy approximation bounds. To provide the
theoretical underpinnings for this technique, we develop a corresponding
theory of sparse separable decompositions of functions, including notions of
rank and nuclear norm that specialize to the usual ones for matrices and
operators but apply to also hypermatrices and tensors.

\end{abstract}

\section{Introduction}\label{intro-sec}

\IEEEPARstart{T}{here} are two simple ideas for reducing the complexity or
dimension of a problem that are widely applicable because of their simplicity
and generality:

\begin{itemize}
\item \textbf{Sparsity:} resolving a complicated entity, represented by a
function $f$, into a sum of a small number of simple or elemental
constituents:
\[
f=\sum_{p=1}^{r}\alpha_{p}g_{p}.
\]

\item \textbf{Separability:} decoupling a complicated entity, represented by a
function $g$, that depends on multiple factors into a product of simpler
constituents, each depending only on one factor:
\[
g(\mathbf{x}_{1},\dots,\mathbf{x}_{d})=\prod_{k=1}^{d}\varphi_{k}%
(\mathbf{x}_{k}).
\]

\end{itemize}

The two ideas underlie some of the most useful techniques in engineering and
science --- Fourier, wavelets, and other orthogonal or sparse representations
of signals and images, singular value and eigenvalue decompositions of
matrices, separation-of-variables, Fast Fourier Transform, mean field
approximation, etc. This article examines the model that combines these two
simple ideas:%
\begin{equation}
f(\mathbf{x}_{1},\dots,\mathbf{x}_{d})=\sum_{p=1}^{r}\alpha_{p}\prod_{k=1}%
^{d}\varphi_{kp}(\mathbf{x}_{k}), \label{fund}%
\end{equation}
and we are primarily interested in its \textit{inverse problem},
i.e., identification of the factors $\varphi_{kp}$ based on noisy measurements
of $f$. We shall see that this is a surprisingly effective method for a wide
range of identification problems.

Here, $f$ is approximately encoded by $r$ scalars, $\boldsymbol{\alpha}%
=(\alpha_{1},\dots,\alpha_{r})\in\mathbb{C}^{r}$, and $dr$ functions,
$\varphi_{kp}$, $k=1,\dots,d$; $p=1,\dots,r$. Since $d$ and $r$ are both
assumed to be small, we expect \eqref{fund} to be a very compact, possibly
approximate, representation of $f$. We will assume that all these functions
live in a Hilbert space with inner product $\langle \cdot\, , \cdot \rangle$, and that $\varphi_{kp}$ are of unit norm (clearly
possible since the norm of $\varphi_{kp}$ can be `absorbed into' the
coefficient $\alpha_{p}$ in \eqref{fund}).

Let $\mu_{k}=\max_{p\neq q}\lvert\langle\varphi_{kp},\varphi_{kq}\rangle
\rvert$ and define the \textit{relative incoherence} $\omega_{k}=(1-\mu
_{k})/\mu_{k}$ for $k=1,\dots,d$. Note that $\mu_{k}\in[0,1]$ and $\omega
_{k}\in[0,\infty]$. We will show that if $d\geq3$, and
\begin{equation}
\sum_{k=1}^{d}\omega_{k}\geq2r-1, \label{unique}%
\end{equation}
then the decomposition in \eqref{fund} is essentially \textit{unique} and
\textit{sparsest} possible, i.e., $r$ is minimal. Hence we may in principle
identify $\varphi_{kp}$ based only on measurements of the mixture $f$.

One of the keys in the identifiability requirement is that $d\geq3$ or
otherwise (when $d=1$ or $2$) the result would not hold. We will show that the
condition $d\geq3$ however leads to a difficulty (that does not happen when
$d=1$ or $2$). Since it is rarely, if not never, the case that one has the
exact values of $f$, the decomposition \eqref{fund} is only useful in an
idealized scenario. In reality, one has $\hat{f}=f+\varepsilon$, an estimate
of $f$ corrupted by noise $\varepsilon$. Solving the inverse problem to
\eqref{fund} would require that we solve a best approximation problem. For
example, with the appropriate noise models (see Section~\ref{est}), the best
approximation problem often takes the form%
\begin{equation}
\operatorname*{argmin}_{\boldsymbol{\alpha}\in\mathbb{C}^{r},\;\lVert
\varphi_{kp}\rVert=1}\left\Vert \hat{f}-\sum_{p=1}^{r}\alpha_{p}\prod
_{k=1}^{d}\varphi_{kp}\right\Vert , \label{approx}%
\end{equation}
with $\lVert\,\cdot\,\rVert$ an $L^{2}$-norm. Now, the trouble is that when
$d\geq3$, this best approximation problem may not have a solution --- because
the infimum of the loss function is unattainable in general, as we will
discuss in Section~\ref{sec:exist}. In view of this, our next result is that
when
\begin{equation}
\prod_{k=1}^{d}(1+\omega_{k})>r, \label{exist}%
\end{equation}
the infimum in \eqref{approx} is always attainable, thereby alleviating the
aforementioned difficulty. A condition that meets both \eqref{unique} and
\eqref{exist} follows from the arithmetic-geometric mean
inequality
\[
\left[  \prod_{k=1}^{d}(1+\omega_{k})\right]  ^{1/d}\leq1+\frac{1}{d}%
\sum_{k=1}^{d}\omega_{k}.
\]

\section{Sparse separable decompositions\label{sepdecomp}}

The notion of \textit{sparsity} dates back to harmonic analysis \cite{Stein,
Young, Mallat} and approximation theory \cite{Temlyakov}, and has received a
lot of recent attention in compressive sensing
\cite{Dono06:it,DonoH01:ieeeit,CoheDD09:jams}. The notion of
\textit{separability} is also classical --- the basis behind the
separation-of-variables technique in partial differential equations
\cite{svpde} and special functions \cite{svbook}, fast Fourier transforms on
arbitrary groups \cite{svfft}, mean field approximations in statistical
physics \cite{StatPhy}, and the na\"{\i}ve Bayes model in machine learning
\cite{Bishop,LimC09:jchemo}. We describe a simple model that incorporates the
two notions.

The function $f:X\rightarrow\mathbb{C}$ or $\mathbb{R}$ to be resolved into
simpler entities will be referred to as our \textit{target function}. We will
treat the discrete ($X$ is finite or countably infinite) and continuous ($X$
is a continuum) cases on an equal footing. The discrete cases are when $f$ is
a vector (if $X=[n_{1}] :=\{1,\dots,n_{1}\}$), a matrix (if $X=[n_{1}]\times[
n_{2}]$), a hypermatrix (if $X=[n_{1}]\times[ n_{2}]\times\dots\times[ n_{d}%
]$), while the usual continuous cases are when $f$ is a function on some
domain $X=\Omega\subseteq\mathbb{R}^{m}$ or $\mathbb{C}^{m}$. In the discrete
cases, the set of target functions under consideration are identified with
$\mathbb{C}^{n_{1}}$, $\mathbb{C}^{n_{1}\times n_{2}}$, $\mathbb{C}%
^{n_{1}\times n_{2}\times\dots\times n_{d}}$ respectively whereas in the
continuous cases, we usually impose some additional regularity structures such
integrability or differentability, so that the set of target functions under
consideration are $L^{2}(\Omega)$ or $C^{\infty}(\Omega)$ or $H^{k}%
(\Omega)=W^{k,2}(\Omega)$, etc. We will only assume that the space of target
functions is a Hilbert space. Note that the requirement $d\geq3$ implies that
$f$ is at least a $3$-dimensional hypermatrix in the discrete case or a function
of at least three continuous variables, i.e., $m\geq3$, in the continuous
case. The identifiability does not work for (usual $2$-dimensional) matrices
or bivariate functions. With \eqref{fund} in mind, we will call $f$ a
$d$\textit{-partite} or \textit{multipartite} function if we wish to partition
its arguments into $d$ blocks of variables.

We will briefly examine the decompositions and approximations of our target
function into a sum or integral of separable functions, adopting a tripartite
notation for simplicity. There are three cases:

\begin{itemize}
\item \textbf{Continuous:}%
\begin{equation}
f(\mathbf{x},\mathbf{y},\mathbf{z})=\int_{T}\theta(\mathbf{x},\mathbf{t}%
)\varphi(\mathbf{y},\mathbf{t})\psi(\mathbf{z},\mathbf{t})\,d\nu(\mathbf{t}).
\label{cont}%
\end{equation}
Here, we assume that $\nu$ is some given Borel measure and that $T$ is compact.

\item \textbf{Semidiscrete:}%
\begin{equation}
f(\mathbf{x},\mathbf{y},\mathbf{z})=\sum_{p=1}^{r}\theta_{p}(\mathbf{x}%
)\varphi_{p}(\mathbf{y})\psi_{p}(\mathbf{z}). \label{semi}%
\end{equation}
This may be viewed as a discretization of the continuous case in the
$\mathbf{t}$ variable, i.e., $\theta_{p}(\mathbf{x})=\theta(\mathbf{x}%
,\mathbf{t}_{p})$, $\varphi_{p}(\mathbf{y})=\varphi(\mathbf{y},\mathbf{t}%
_{p})$, $\psi_{p}(\mathbf{z})=\psi(\mathbf{z},\mathbf{t}_{p})$.

\item \textbf{Discrete:}%
\begin{equation}
a_{ijk}=\sum_{p=1}^{r}u_{ip}v_{jp}w_{kp}. \label{disc}%
\end{equation}
This may be viewed as a further discretization of the semidiscrete case,
i.e.$~a_{ijk}=f(\mathbf{x}_{i},\mathbf{y}_{j},\mathbf{z}_{k})$, $u_{ip}%
=\theta_{p}(\mathbf{x}_{i})$, $v_{jp}=\varphi_{p}(\mathbf{y}_{j})$,
$w_{kp}=\psi_{p}(\mathbf{z}_{k})$.
\end{itemize}

It is clear that when $i,j,k$ take finitely many values, the discrete
decomposition \eqref{disc} is always possible with a finite $r$ since the
space is of finite dimension. If $i,j,k$ could take infinitely many values,
then the finiteness of $r$ requires that equality be replaced by approximation
to any arbitrary precision $\varepsilon>0$ in some suitable norm. This follows
from the following observation about the semidiscrete decomposition: The space
of functions with a semidiscrete representation as in \eqref{semi}, with $r$
finite, is dense in $C^{0}(\Omega)$, the space of continuous functions. This
is just a consequence of the Stone-Weierstrass theorem \cite{Cybe89:mcss}.
Discussion of the most general case \eqref{cont} would require us to go into
integral operators, which we will not do as in the present framework we are
interested in applications that rely on the inverse problems corresponding to
\eqref{semi} and \eqref{disc}. Nonetheless \eqref{cont} is expected to be
useful and we state it here for completeness. Henceforth, we will simply refer
to \eqref{semi} or \eqref{disc} as a \textit{multilinear decomposition}, by
which we mean a decomposition into a \textit{linear combination of separable
functions}. We note here that the finite-dimensional discrete version has been
studied under several different names --- see Section~\ref{applications-sec}.
Our emphasis in this paper is the semidiscrete version \eqref{semi} that
applies to multipartite functions on arbitrary domains and are not necessarily
finite-dimensional. As such, we will frame most of our discussions in
terms of the semidiscrete case, which of course includes the discrete
version \eqref{disc} as a special case (when $\mathbf{x},\mathbf{y}%
,\mathbf{z}$ take only finite discrete values).

\begin{example}[Mixture of Gaussians]
Multilinear decompositions arise in many contexts. In machine
learning or nonparametric statistics, a fact of note is that Gaussians are
separable
\[
\exp(x^{2}+y^{2}+z^{2})=\exp(x^{2})\exp(y^{2})\exp(z^{2}).
\]
More generally for symmetric positive-definite $A\in\mathbb{R}^{n\times n}$
with eigenvalues $\Lambda=\operatorname{diag}(\lambda_{1},\dots,\lambda_{n}%
)$,
\[
\exp(\mathbf{x}^{\mathsf{T}}A\mathbf{x})=\exp(\mathbf{z}^{\mathsf{T}}%
\Lambda\mathbf{z})=\prod_{i=1}^{n}\exp(\lambda_{i}z_{i}^{2}),
\]
under a linear change of coordinates $\mathbf{z}=Q^{\mathsf{T}}\mathbf{x}$
where $A=Q\Lambda Q^{\mathsf{T}}$. Hence Gaussian mixture models of the form
\[
f(\mathbf{x})=\sum_{j=1}^{m}\alpha_{j}\exp[(\mathbf{x}-\boldsymbol{\mu}%
_{j})^{\mathsf{T}}A_{j}(\mathbf{x}-\boldsymbol{\mu}_{j})],
\]
where $A_{i}A_{j}=A_{j}A_{i}$ for all $i\neq j$ (and therefore $A_{1}%
,\dots,A_{m}$ have a common eigenbasis) may likewise be transformed with a
suitable linear change of coordinates into a multilinear decomposition as in \eqref{semi}.
\end{example}

We will later see several more examples from signal processing,
telecommunications, and spectroscopy.

\subsection{Modeling}

The multilinear decomposition --- an \textit{additive} decomposition into
\textit{multiplicatively} decomposable components --- is extremely simple but
models a wide range of phenomena in signal processing and spectroscopy. The
main message of this article is that the corresponding inverse problem ---
recovering the factors $\theta_{p},\varphi_{p},\psi_{p}$ from noisy
measurements of $f$ --- can be solved under mild assumptions and yields a
class of techniques for a range of applications
(cf.\ Section~\ref{applications-sec}) that we shall collectively call
\textit{multilinear identification}. We wish to highlight in particular that
multilinear identification gives a deterministic approach for solving the
problem of joint localization and estimation of radiating sources with short
data lengths. This is superior to previous cumulants-based approaches \cite{ComoJ10}, which require
(i) longer data lengths; and (ii) statistically independent sources.

The experienced reader would probably guess that such a powerful technique
must be fraught with difficulties and he would be right. The inverse problem
to \eqref{semi}, like most other inverse problems, faces issues of existence,
uniqueness, and computability. The approximation problem involved can be
ill-posed in the worst possible way (cf.\ Section~\ref{problem-sec}).
Fortunately, in part prompted by recent work in
compressed sensing \cite{CandR07:ip,CoheDD09:jams, Dono06:it,
DonoE03:pnas, GribN03:it} and matrix completion \cite{CR1,CandT10:it, FHB, Gross}), we
show that mild assumptions on \textit{coherence} allows one to overcome
most of these difficulties (cf.\ Section~\ref{angular-sec}).

\section{Finite rank multipartite functions\label{problem-sec}}

In this section, we will discuss the notion of rank, which measures the
sparsity of a multilinear decomposition, and the notion of Kruskal rank, which
measures the uniqueness of a multilinear decomposition in a somewhat more
restrictive sense. Why is uniqueness important? It can be answered in one
word: Identifiability. More specifically, a unique decomposition means that we
may in principle identify the factors. To be completely precise, we will first
need to define the terms in the previous sentence, namely, `unique',
`decomposition', and `factor'. Before we do that, we will introduce the tensor
product notation. It is not necessary to know anything about tensor product of
Hilbert spaces to follow what we present below. We shall assume that all our
Hilbert spaces are separable and so there is no loss of generality in assuming
at the outset that they are just $L^{2}(X)$ for some $\sigma$-finite $X$.

Let $X_{1},\dots,X_{d}$ be $\sigma$-finite measurable spaces. There is a
natural Hilbert space isomorphism
\begin{equation}
L^{2}(X_{1}\times\dots\times X_{d})\cong L^{2}(X_{1})\otimes\dots\otimes
L^{2}(X_{d}). \label{iso}%
\end{equation}
In other words, every $d$-partite $L^{2}$-function $f:X_{1}\times\dots\times
X_{d}\rightarrow\mathbb{C}$ may be expressed as\footnote{Point values of
$L^{p}$-functions are undefined in general. So equations like these are taken
to implicitly mean \textit{almost everywhere}. Anyway all functions that arise
in our applications will at least be continuous and so this is really not a
point of great concern.}
\begin{equation}
f(\mathbf{x}_{1},\dots,\mathbf{x}_{d})=\sum_{p=1}^{\infty}\varphi
_{1p}(\mathbf{x}_{1})\cdots\varphi_{dp}(\mathbf{x}_{d}), \label{sep1}%
\end{equation}
with $\varphi_{kp}\in L^{2}(X_{k})$. The \textit{tensor product} of functions
$\varphi_{1}\in L^{2}(X_{1}),\dots,\varphi_{d}\in L^{2}(X_{d})$ is denoted by
$\varphi_{1}\otimes\dots\otimes\varphi_{d}$ and is the function in
$L^{2}(X_{1}\times\dots\times X_{d})$ defined by%
\[
\varphi_{1}\otimes\dots\otimes\varphi_{d}(\mathbf{x}_{1},\dots,\mathbf{x}%
_{d})=\varphi_{1}(\mathbf{x}_{1})\cdots\varphi_{d}(\mathbf{x}_{d}).
\]
With this notation, we may rewrite \eqref{sep1} as%
\[
f=\sum_{p=1}^{\infty}\varphi_{1p}\otimes\dots\otimes\varphi_{dp}.
\]
A point worth noting here is that:

\begin{quote}
``\textit{Multipartite functions are infinite-dimensional
tensors.}''
\end{quote}

Finite-dimensional tensors are simply the special case when $X_{1},\dots
,X_{d}$ are all finite sets (see Example~\ref{eg:hitch}). In particular, a
multivariate function\footnote{We clarify our terminologies: A multipartite
function is one for which the arguments $\mathbf{x}_{1},\dots,\mathbf{x}_{d}$
can come from any $X_{1},\dots,X_{d}$ but a multivariate function, in the
usual sense of the word, is one where $X_{1},\dots,X_{d}$ are (measurable)
subsets of $\mathbb{R}$.  For example, while
\[
g(u,v,w,x,y,z)=\varphi_{1}(u,v)\varphi_{2}(w)\varphi_{3}(x,y,z)
\]
is not separable in the multivariate sense, it is separable in the multipartite sense:
for $\mathbf{x}_{1}=(u,v)$, $\mathbf{x}_{2}=w$, $\mathbf{x}_{3}=(x,y,z)$,
\[
g(\mathbf{x}_{1},\mathbf{x}_{2},\mathbf{x}_{3})=\varphi_{1}(\mathbf{x}%
_{1})\varphi_{2}(\mathbf{x}_{2})\varphi_{3}(\mathbf{x}_{3}).
\]} $f\in
L^{2}(\mathbb{R}^{d})$ is a an infinite-dimensional tensor that can expressed
as an infinite sum of a tensor product of $\varphi_{1p},\dots,\varphi_{dp}\in
L^{2}(\mathbb{R})$ and $L^{2}(\mathbb{R}^{d})\cong L^{2}(\mathbb{R}%
)\otimes\dots\otimes L^{2}(\mathbb{R)}$. We shall have more to say about this
later in conjunction with Kolmogorov's superposition principle for
multivariate functions.

In this paper, functions having a \textit{finite} decomposition will play a
central role; for these we define
\begin{equation}
\operatorname*{rank}(f):=\min\left\{  r\in\mathbb{N}:f=\sum_{p=1}^{r}%
\varphi_{1p}\otimes\dots\otimes\varphi_{dp}\right\}  \label{rank-eq}%
\end{equation}
provided $f\neq0$. The zero function is defined to have rank $0$ and we say
$\operatorname{rank}(f)=\infty$ if such a decomposition is not possible.

We will call a function $f$ with $\operatorname{rank}(f)\leq r$ a
\textit{rank-}$r$\textit{ function}. Such a function may be written as a sum
of $r$ separable functions but possibly fewer. A decomposition of the form%
\begin{equation}
f=\sum_{p=1}^{r}\varphi_{1p}\otimes\dots\otimes\varphi_{dp} \label{add-eq}%
\end{equation}
will be called  a \textit{rank-}$r$\textit{ multilinear decomposition}. Note
that the qualificative `rank-$r$' will always mean `rank not more than $r$'.
If we wish to refer to a function $f$ with rank exactly $r$, we will just
specify that $\operatorname*{rank}(f)=r$. In this case, the rank-$r$
multilinear decomposition in \eqref{add-eq} is of mininum length and we call
it a \textit{rank-retaining multilinear decomposition} of $f$.

A rank-$1$ function is both non-zero and decomposable, i.e., of the form
$\varphi_{1}\otimes\dots\otimes\varphi_{d}$ where $\varphi_{k}\in L^{2}%
(X_{k})$. This agrees precisely with the notion of a separable function.
Observe that the inner product (and therefore the norm) on $L^{2}(X_{1}%
\times\dots\times X_{d})$ of a rank-$1$ function splits into a product%
\begin{equation}
\langle\varphi_{1}\otimes\dots\otimes\varphi_{d},\psi_{1}\otimes\dots
\otimes\psi_{d}\rangle=\langle\varphi_{1},\psi_{1}\rangle_{1}\cdots
\langle\varphi_{d},\psi_{d}\rangle_{d} \label{ip-eq}%
\end{equation}
where $\langle\cdot,\cdot\rangle_{p}$ denotes the inner product of
$L^{2}(X_{p})$. This inner product extends linearly to finite-rank elements of
$L^{2}(X_{1}\times\dots\times X_{d})$: for $f=\sum_{p=1}^{r}\varphi
_{1p}\otimes\dots\otimes\varphi_{dp}$ and $g=\sum_{q=1}^{s}\psi_{1q}%
\otimes\dots\otimes\psi_{dq}$, we have%
\[
\langle f,g\rangle=\sum_{p,q=1}^{r,s}\langle\varphi_{1p},\psi_{1q}\rangle
_{1}\cdots\langle\varphi_{dp},\psi_{dq}\rangle_{d}.
\]
In fact this is how a tensor product of Hilbert spaces (the right hand side of
\eqref{iso}) is usually defined, namely, as the completion of the set of
finite-rank elements of $L^{2}(X_{1}\times\dots\times X_{d})$ under this inner product.

When $X_{1},\dots,X_{d}$ are finite sets, then all functions in $L^{2}%
(X_{1}\times\dots\times X_{d})$ are of finite rank (and may in fact be viewed
as hypermatrices or tensors as discussed in Section~\ref{sepdecomp}).
Otherwise there will be functions in $L^{2}(X_{1}\times\dots\times X_{d})$ of
infinite rank. However, since we have assumed that $X_{1},\dots,X_{d}$ are
$\sigma$-finite measurable spaces, the set of all finite-rank $f$ will always
be dense in $L^{2}(X_{1}\times\dots\times X_{d})$ by the Stone-Weierstrass theorem.

The next statement is a straightforward observation about multilinear
decompositions of finite-rank functions but since it is central to
this article we state it as a theorem. It is also tempting to call the
decomposition a `singular value decomposition' given its similarities with the
usual matrix singular value decomposition (cf.~Example~\ref{eg:svd}).

\begin{theorem}
[`Singular value decomposition' for multipartite functions]\label{thm:svd}Let
$f\in L^{2}(X_{1}\times\dots\times X_{d})$ be of finite rank. Then there
exists a rank-$r$ multilinear decomposition
\begin{equation}
f=\sum_{p=1}^{r}\sigma_{p}\varphi_{1p}\otimes\dots\otimes\varphi_{dp}
\label{CP-eq}%
\end{equation}
such that%
\begin{equation}
r=\operatorname*{rank}(f), \label{rank}%
\end{equation}
the functions $\varphi_{kp}\in L^{2}(X_{p})$ are of unit norm,%
\begin{equation}
\lVert\varphi_{kp}\rVert=1\quad\text{for all }k=1,\dots,d,\quad p=1,\dots,r,
\label{unit}%
\end{equation}
the coefficients $\sigma_{1},\dots,\sigma_{r}$ are real positive, and%
\begin{equation}
\sigma_{1}\geq\sigma_{2}\geq\dots\geq\sigma_{r}>0. \label{descend}%
\end{equation}

\end{theorem}

\begin{proof}
This requires nothing more than rewriting the sum in \eqref{add-eq} as a
linear combination with the positive $\sigma_{p}$'s accounting for the norms of
the summands and then re-indexing them in descending order of magnitudes.
\end{proof}

While the usual singular value decomposition of a matrix would also have
properties \eqref{rank}, \eqref{unit}, and \eqref{descend}, the one crucial
difference here is that our `singular vectors' $\varphi_{k1},\dots
,\varphi_{kr}$ in \eqref{CP-eq} will only be of unit norms but will not in
general be orthonormal. Given this, we will not expect properties like the
Eckhart-Young theorem, or that $\sigma_{1}^{2}+\dots+\sigma_{r}^{2}=\lVert
f\rVert^{2}$, etc, to hold for \eqref{CP-eq} (cf.\ Section~\ref{illposed-sec}
for more details).

One may think of the multilinear decomposition \eqref{CP-eq} as being similar
in spirit, although not in substance, to Kolmogorov's superposition principle
\cite{Kolmogorov}; the main message of which is that:

\begin{quote}
``\textit{There are no true multivariate functions.}''
\end{quote}

More precisely, the principle states that continuous functions in multiple
variables can be expressed as a composition of a univariate function with
other univariate functions. For readers not familiar with this remarkable
result, we state here a version of it due to Kahane \cite{Kahane}

\begin{theorem}
[Kolmogorov superposition]\label{thm:Kolmo}Let $f:[0,1]^{d}\rightarrow
\mathbb{R}$ be continuous. Then there exist constants $\lambda_{1}%
,\dots,\lambda_{d}\in\mathbb{R}$ and Lipschitz continuous functions
$\varphi_{1},\dots,\varphi_{d}:[0,1]\rightarrow[0,1]$ such that%
\[
f(x_{1},\dots,x_{d})=\sum_{p=1}^{2d+1}g(\lambda_{1}\varphi_{p}(x_{1}%
)+\dots+\lambda_{d}\varphi_{p}(x_{d})).
\]

\end{theorem}

It is in general not easy to determine $g$ and $\varphi_{1},\dots
,\varphi_{2d+1}$ given such a function $f$. A multilinear decomposition of the
form \eqref{CP-eq} alleviates this by allowing $g$ to be the simplest
multivariate function, namely, the product function,%
\begin{equation}
g(t_{1},\dots,t_{d})=t_{1}t_{2}\cdots t_{d}, \label{eq:product}%
\end{equation}
and unlike the univariate $g$ in Theorem~\ref{thm:Kolmo}, the $g$ in
\eqref{eq:product} works universally for any function $f$ --- only the
$\varphi_{p}$'s need to be constructed. Furthermore, \eqref{CP-eq} applies
more generally to functions on a product of general domains $X_{1},\dots
,X_{d}$ whereas Theorem~\ref{thm:Kolmo} only applies if they are compact
intervals of $\mathbb{R}$.

At this stage, it would be instructive to give a few examples for concreteness.

\begin{example}[Singular value decomposition]
\label{eg:svd}Let $A\in\mathbb{C}^{m\times n}$ be a matrix of rank $r$. Then
it can be decomposed in infinitely many ways into a sum of rank-$1$ terms as%
\begin{equation}
A=\sum_{p=1}^{r}\sigma_{p}\mathbf{u}_{p}\mathbf{v}_{p}^{\ast} \label{svd}%
\end{equation}
where $\mathbf{u}_{p}\in\mathbb{C}^{m}$ and $\mathbf{v}_{p}\in\mathbb{C}^{n}$
are unit-norm vectors and $\sigma_{1}\geq\dots\geq\sigma_{r}>0$. Note that if
we regard $A$ as a complex-valued function on its row and column indices $i$
and $j$ as described earlier in Section~\ref{sepdecomp}, then \eqref{svd} may
be written as%
\[
a(i,j)=\sum_{p=1}^{r}\sigma_{p}u_{p}(i)v_{p}(j),
\]
which clearly is the same as \eqref{sep1}. The singular value decomposition
(\textsc{svd}) of $A$ yields one such decomposition, where $\{\mathbf{u}%
_{1},\dots,\mathbf{u}_{r}\}$ and $\{\mathbf{v}_{1},\dots,\mathbf{v}_{r}\}$ are
both orthonormal. But in general a rank-retaining decomposition of the form
\eqref{CP-eq} will not have such a property.
\end{example}

\begin{example}[Schmidt decomposition]
\label{eg:schmidt}The previous example can be generalized to infinite
dimensions. Let $A:\mathbb{H}_{1}\rightarrow\mathbb{H}_{2}$ be a compact
operator (also known as a completely continuous operator) between two
separable Hilbert spaces. Then the Schmidt decomposition theorem says that
there exist orthonormal basis $\{\varphi_{p}\in\mathbb{H}_{2}:p\in
\mathbb{N}\}$ and $\{\psi_{p}\in\mathbb{H}_{1}:p\in\mathbb{N}\}$ so that%
\begin{equation}
Af=\sum_{p=1}^{\infty}\sigma_{p}\langle\psi_{p},f\rangle\varphi_{p}
\label{schmidt}%
\end{equation}
for every $f\in\mathbb{H}_{1}$. In tensor product notation, \eqref{schmidt}
becomes%
\[
A=\sum_{p=1}^{\infty}\sigma_{p}\varphi_{p}\otimes\psi_{p}^{\ast}.
\]
where $\psi_{p}^{\ast}$ denotes the dual form of $\psi_{p}$.
\end{example}

Examples~\ref{eg:svd} and \ref{eg:schmidt} are well-known but they are
bipartite examples, i.e.~$d=2$ in \eqref{CP-eq}. This article is primarily
concerned with the $d$-partite case where $d\geq3$, which has received far
less attention. As we have alluded to in the previous section, the
identification techniques in this article will rely crucially on the fact that
$d\geq3$.

\begin{example}
\label{eg:hitch}Let $A\in\mathbb{C}^{l\times m\times n}$ be a $3$-dimensional
hypermatrix. The outer product of three vectors $\mathbf{u}\in\mathbb{C}^{l}$,
$\mathbf{v}\in\mathbb{C}^{m}$, $\mathbf{w}\in\mathbb{C}^{n}$ is defined by%
\[
\mathbf{u}\otimes\mathbf{v}\otimes\mathbf{w}=(u_{i}v_{j}w_{k})_{i,j,k=1}%
^{l,m,n}\in\mathbb{C}^{l\times m\times n}.
\]
The rank of $A$ is defined to be the minimum $r\in\mathbb{N}$ such that $A$
can be written in the form%
\begin{equation}
A=\sum_{p=1}^{r}\sigma_{p}\mathbf{u}_{p}\otimes\mathbf{v}_{p}\otimes
\mathbf{w}_{p}, \label{eq5}%
\end{equation}
and if $A=0$, then its rank is set to be $0$. This agrees of course with our
use of the word rank in \eqref{rank-eq}, the only difference is notational,
since \eqref{eq5} may be written in the form
\[
a(i,j,k)=\sum_{p=1}^{r}\sigma_{p}u_{p}(i)v_{p}(j)w_{p}(k).
\]
This definition of rank is invariant under the natural action\footnote{$\operatorname{GL}_n(\mathbb{C}) :=\{ A \in \mathbb{C}^{n \times n}: \det (A) \ne 0\}$ denotes
the general linear goup of nonsingular $n \times n$ complex matrices.}
$\operatorname{GL}_l(\mathbb{C})\times\operatorname{GL}_m(\mathbb{C})\times\operatorname{GL}_n(\mathbb{C})$ on
$\mathbb{C}^{l\times m\times n}$ \cite[Lemma~2.3]{DesiL08:simax}, i.e., for
any $X\in\operatorname{GL}_l(\mathbb{C}),Y\in\operatorname{GL}_m(\mathbb{C}),Z\in\operatorname{GL}_n(\mathbb{C})$,%
\begin{equation}
\operatorname*{rank}((X,Y,Z)\cdot A)=\operatorname*{rank}(A). \label{eq7}%
\end{equation}
The definition also extends easily to $d$-dimensional hypermatrices in
$\mathbb{C}^{n_{1}\times\dots\times n_{d}}$ and when $d=2$ reduces to the
usual definition in Example~\ref{eg:svd} for matrix rank. This definition is
due to F.~L.~Hitchcock \cite{Hitc27:jmp} and is often called \textit{tensor
rank}. The only difference here is that our observation in
Theorem~\ref{thm:svd} allows us to impose the conditions%
\[
\sigma_{1}\geq\sigma_{2}\geq\dots\geq\sigma_{r}%
\]
and%
\begin{equation}
\lVert\mathbf{u}_{p}\rVert=\lVert\mathbf{v}_{p}\rVert=\lVert\mathbf{w}%
_{p}\rVert=1,\quad p=1,\dots,r, \label{eq6}%
\end{equation}
while leaving $\operatorname*{rank}(A)$ unchanged, thus bringing \eqref{eq5}
closer in form to its matrix cousin \eqref{svd}. What is lost here is that the
sets $\{\mathbf{u}_{1},\dots,\mathbf{u}_{r}\},\{\mathbf{v}_{1},\dots
,\mathbf{v}_{r}\},\{\mathbf{w}_{1},\dots,\mathbf{w}_{r}\}$ can no longer be
chosen to be orthonormal as in Example~\ref{eg:svd}, the unit norm condition
\eqref{eq6} is as far as we may go. In fact for a generic $A\in\mathbb{C}%
^{l\times m\times n}$, we will always have%
\[
r>\max(l,m,n),
\]
and $\{\mathbf{u}_{1},\dots,\mathbf{u}_{r}\},\{\mathbf{v}_{1},\dots
,\mathbf{v}_{r}\},\{\mathbf{w}_{1},\dots,\mathbf{w}_{r}\}$ will be
overcomplete sets in $\mathbb{C}^{l},\mathbb{C}^{m},\mathbb{C}^{n}$ respectively.
\end{example}

Perhaps it is worthwhile saying a word concerning our use of the words
`tensor' and `hypermatrix': A $d$-tensor or order-$d$ tensor is an element of
a tensor product of $d$ vector spaces $\mathbb{V}_{1}\otimes\dots
\otimes\mathbb{V}_{d}$; a $d$-dimensional hypermatrix is an element of
$\mathbb{C}^{n_{1}\times\dots\times n_{d}}$. If we choose bases on
$\mathbb{V}_{1},\dots,\mathbb{V}_{d}$, then any $d$-tensor $\mathbf{A}%
\in\mathbb{V}_{1}\otimes\dots\otimes\mathbb{V}_{d}$ will have a unique
coordinate representation as a $d$-dimensional hypermatrix $A\in
\mathbb{C}^{n_{1}\times\dots\times n_{d}}$, where $n_{k}=\dim(\mathbb{V}_{k}%
)$. A notion defined on a hypermatrix is only defined on the tensor (that is
represented in coordinates by the hypermatrix) if that notion is independent
of the choice of bases. So the use of the word `tensor rank' is in fact well
justified because of \eqref{eq7}. For more details, we refer the reader to
\cite{HLA}.

\section{Uniqueness of multilinear decompositions\label{sec:unique}}

In Theorem~\ref{thm:svd}, we chose the coefficients to be in descending order
of magnitude and require the factors in each separable term to be of unit
norm. This is largely to ensure as much uniqueness in the multilinear
decomposition as generally possible. However there remain two obvious ways to
obtain trivially different multilinear decompositions: (i) one may scale the
factors $\varphi_{1p},\dots,\varphi_{dp}$ by arbitrary unimodulus complex
numbers as long as their product is $1$; (ii) when two or more successive
coefficients are equal, their orders in the sum may be arbitrarily permuted.
We will call a multilinear decomposition of $f$ that meets the conditions in
Theorem~\ref{thm:svd} \textit{essentially unique} if the only other such
decompositions of $f$ differ in one or both of these manners.

It is perhaps astonishing that when $d>2$, a sufficient condition for
essential uniqueness can be derived with relatively mild conditions on the
factors. This relies on the notion of Kruskal rank, which we will now define.

\begin{definition}
\label{kruskal-def} Let $\Phi=\{\varphi_{1},\dots,\varphi_{r}\}$ be a finite
collection of vectors of unit norm in $L^{2}(X_{1}\times\dots\times X_{d})$.
The \textbf{Kruskal rank} of $\Phi$, denoted $\operatorname{krank}\Phi$, is
the largest $k\in\mathbb{N}$ so that every $k$-element subset of $\Phi$
contains linearly independent elements.
\end{definition}

This notion was originally introduced in \cite{Krus77:laa}. It is related to
the notion of \textit{spark} introduced in compressed sensing
\cite{DonoE03:pnas,GribN03:it}, defined as the smallest $k\in\mathbb{N}$ so
that there is at least one $k$-element subset of $\Phi$ that is linearly
dependent. The relation is simple to describe, $\operatorname{spark}%
\Phi=\operatorname{krank}\Phi+1$, and it follows immediately from the
respective definitions. It is clear that $\dim\operatorname{span}\Phi
\geq\operatorname{krank} \Phi$.

We now generalize Kruskal's famous result \cite{Krus77:laa,SB} to tensor
products of arbitrary Hilbert spaces, possibly of infinite dimensions. But
first let us be more specific about essential uniqueness.

\begin{definition}
We shall say that a multilinear decomposition of the form \eqref{CP-eq}
(satisfying both \eqref{descend} and \eqref{unit}) is \textbf{essentially
unique} if given another such decomposition,%
\[
\sum_{p=1}^{r}\sigma_{p}\varphi_{1p}\otimes\dots\otimes\varphi_{dp}%
=f=\sum_{p=1}^{r}\lambda_{p}\psi_{1p}\otimes\dots\otimes\psi_{dp},
\]
we must have (i) the coefficients $\sigma_{p}=\lambda_{p}$ for all
$p=1,\dots,r$; and (ii) the factors $\varphi_{1p},\dots,\varphi_{dp}$ and
$\psi_{1p},\dots,\psi_{dp}$ differ at most via unimodulus scaling, i.e.%
\begin{equation}
\varphi_{1p}=e^{i\theta_{1p}}\psi_{1p},\dots,\varphi_{dp}=e^{i\theta_{dp}}%
\psi_{dp} \label{unimodulus-def}%
\end{equation}
where $\theta_{1p}+\dots+\theta_{dp}\equiv0\operatorname{mod}2\pi$, for all
$p=1,\dots,r$. In the event when successive coefficients are equal,
$\sigma_{p-1}>\sigma_{p}=\sigma_{p+1}=\dots=\sigma_{p+q}>\sigma_{p+q+1}$, the
uniqueness of the factors in (ii) is only up to relabelling of indices,
i.e.$\ p,\dots,p+q$.
\end{definition}

\begin{lemma}[Infinite-dimensional Kruskal uniqueness]
\label{Kruskal-lem} Let $f\in L^{2}(X_{1}\times\dots\times X_{d})$ be of finite rank. Then a
multilinear decomposition of the form%
\begin{equation}
f=\sum_{p=1}^{r}\sigma_{p}\varphi_{1p}\otimes\dots\otimes\varphi_{dp}
\label{krus-decomp}%
\end{equation}
is both essentially unique and rank-retaining, i.e., $r=\operatorname{rank}f$,
if the following condition is satisfied:
\begin{equation}
2r+d-1\leq\sum_{k=1}^{d}\operatorname{krank}\Phi_{k}, \label{Kruskal-eq}%
\end{equation}
where $\Phi_{k}=\{\varphi_{k1},\dots,\varphi_{kr}\}$ for $k=1,\dots,d$.
\end{lemma}

\begin{proof}
Consider the subspaces $\mathbb{V}_{k}=\operatorname{span}(\varphi_{k1}%
,\dots,\varphi_{kr})$ in $L^{2}(X_{k})$ for each $k=1,\dots,d$. Observe that
$f\in\mathbb{V}_{1}\otimes\dots\otimes\mathbb{V}_{d}$. Clearly
$\dim(\mathbb{V}_{k})\leq r$ and so $\dim(\mathbb{V}_{1}\otimes\dots
\otimes\mathbb{V}_{d})\leq r^{d}$. Now, if we could apply Kruskal's uniqueness theorem 
\cite{Krus77:laa} to the finite-dimensional space $\mathbb{V}_{1}%
\otimes\dots\otimes\mathbb{V}_{d}$, then we may immediately deduce both the
uniqueness and rank-retaining property of \eqref{krus-decomp}. However there
is one caveat: We need to show that Kruskal rank does not change under
restriction to a subspace, i.e.$\ $the value of $\operatorname{krank}%
\{\varphi_{k1},\dots,\varphi_{kr}\}$ in \eqref{Kruskal-eq} is the same whether
we regard $\varphi_{k1},\dots,\varphi_{kr}$ as elements of $L^{2}(X_{k})$ or
as elements of the subspace $\mathbb{V}_{k}$. But this just follows from the
simple fact that linear independence has precisely this property, i.e., if
$v_{1},\dots,v_{n}\in\mathbb{U}\subseteq\mathbb{V}$ are linearly independent
in the vector space $\mathbb{V}$, then they are linearly independent in
the subspace $\mathbb{U}$.
\end{proof}

It follows immediately why we usually need $d\geq3$ for identifiability.

\begin{corollary}
A necessary condition for Kruskal's inequality \eqref{Kruskal-eq} to hold is
that $d\geq3$.
\end{corollary}

\begin{proof}
If $d=2$, then $2r+d-1=2r+1>\operatorname{krank}\Phi_{1}+\operatorname{krank}%
\Phi_{2}$ since the Kruskal rank of of $r$ vectors cannot exceed $r$. Likewise
for $d=1$.
\end{proof}

Lemma~\ref{Kruskal-lem} shows that the condition in \eqref{Kruskal-eq} is
sufficient to ensure uniqueness and it is known that the condition is not
necessary. In an appropriate sense, the condition is sharp \cite{Derksen}. We
should note that the version of Lemma~\ref{Kruskal-lem} that we state here for
general $d\geq3$ is due to Sidiropoulos and Bro \cite{SB}. Kruskal's original
version \cite{Krus77:laa} is only for $d=3$.

The main problem with Lemma~\ref{Kruskal-lem} is that the condition
\eqref{Kruskal-eq} is difficult to check since the right-hand side cannot be
readily computed. See Section~\ref{sec:cc} for a discussion. 

Kruskal's result also does not tell us how often multilinear
decompositions are unique. In the event when the sets $X_{1},\dots,X_{d}$ are
finite, $L^{2}(X_{1}\times\dots\times X_{d})\cong\mathbb{C}^{n_{1}\times
\dots\times n_{d}}$ where $n_{1}=\#X_{1},\dots,n_{d}=\#X_{d}$, and there is a
simple result on uniqueness based simply on a dimension count. Note that the
dimension of $L^{2}(X_{1}\times\dots\times X_{d})$ is the product $n_{1}\cdots
n_{d}$ and the number of parameters needed to describe a separable element of
the form $\lambda\varphi_{1}\otimes\dots\otimes\varphi_{d}$ where $\varphi
_{1},\dots,\varphi_{d}$ are of unit norm is $n_{1}+\dots+n_{d}-d+1$ (each
$\varphi_{k}$ requires $n_{k}-1$ parameters because of the unit norm
constraint, the last `$+1$' accounts for the coefficient $\lambda$). We call
the number%
\[
\left\lceil \frac{\prod_{k=1}^{d}n_{k}}{1-d+\sum_{k=1}^{d}n_{k}}\right\rceil
\]
the \textit{expected rank} of $L^{2}(X_{1}\times\dots\times X_{d})$, since it
is heuristically the minimum $r$ expected for a multilinear decomposition \eqref{CP-eq}.

\begin{proposition}
Let the notations be as above. If $f\in L^{2}(X_{1}\times\dots\times X_{d})$
has rank smaller than the expected rank, i.e.%
\[
\operatorname{rank}(f)<\left\lceil \frac{\prod_{k=1}^{d}n_{k}}{1-d+\sum
_{k=1}^{d}n_{k}}\right\rceil ,
\]
then $f$ admits at most a finite number of distinct rank-retaining decompositions.
\end{proposition}

This proposition has been proved in several cases, including symmetric tensors
\cite{ChiaC06:jlms}, but the proof still remains incomplete for tensors of
most general form \cite{CataGG05:jpaa,AboOP09:tams}.

\section{Estimation of multilinear decompositions\label{est}}

In practice we would only have at our disposal $\hat{f}$, a measurement of $f$
corrupted by noise. Recall that our model for $f$ takes the form%
\begin{equation}
f(\mathbf{x}_{1},\dots,\mathbf{x}_{d})=\sum_{p=1}^{r}\alpha_{p}\prod_{k=1}%
^{d}\varphi_{kp}(\mathbf{x}_{k}). \label{model}%
\end{equation}
Then we would often have to solve an approximation problem corresponding to
\eqref{model} of the form%
\begin{equation}
\operatorname*{argmin}_{\boldsymbol{\alpha}\in\mathbb{C}^{r},\;\lVert
\varphi_{kp}\rVert=1}\left\Vert \hat{f}-\sum_{p=1}^{r}\alpha_{p}\prod_{k=1}^{d}\varphi_{kp}\right\Vert , \label{approx1}%
\end{equation}
which we will call a \textit{best rank-}$r$\textit{ approximation problem}. A
solution to \eqref{approx1}, if exists, will be called a best rank-$r$
approximation of $\hat{f}$.

We will give some motivations as to why such an approximation
is reasonable. Assuming that the norm in \eqref{approx1} is the $L^{2}$-norm
and that the factors $\varphi_{kp}$, $p=1,\dots r$ and $k=1,\dots d$, have
been determined in advance and we are just trying to estimate the parameters
$\alpha_{1},\dots,\alpha_{r}$ from $\hat{f}^{(1)},\dots,$ $\hat{f}^{(N)}$ a
finite sample of size $N$ of measurements of $f$ corrupted by noise, then the
solution of the approximation problem in \eqref{approx1} is in fact (i) a
maximum likelihood estimator (\textsc{mle}) if the noise is zero mean
Gaussian, and (ii) a best linear unbiased estimator (\textsc{blue}) if the
noise has zero mean and finite variance. Of course in our identification
problems, the factors $\varphi_{kp}$'s are not known and have to be estimated
too. A probabilistic model in this situation would take us too far afield.
Note that even for the case $d=2$ and where the domain of $f$ $X_{1} \times X_{2}$ is a finite
set, a case that essentially reduces to principal components analysis
(\textsc{pca}), a probabilistic model along the lines of \cite{TipBis} requires
several strong assumptions and was only developed as late as 1999. The lack of
a formal probabilistic model has not stopped \textsc{pca}, proposed in 1901
\cite{pca}, to be an invaluable tool in the intervening century.

\section{Existence of best multilinear approximation\label{illposed-sec}}

As we mentioned in the previous section, in realistic situation where
measurements are corrupted by additive noise, one has to extract the factors
$\varphi_{kp}$'s and $\alpha_{p}$ through solving an approximation problem
\eqref{approx1}, that we now write in a slightly different (but equivalent)
form,%
\begin{equation}
\operatorname*{argmin}_{\boldsymbol{\alpha}\in[0,\infty)^{r},\;\lVert
\varphi_{kp}\rVert=1}\left\Vert \hat{f}-\sum_{p=1}^{r}\alpha_{p}\prod
_{k=1}^{d}\varphi_{kp}\right\Vert . \label{fund1}%
\end{equation}
Note that by Theorem~\ref{thm:svd}, we may assume that the coefficients
$\boldsymbol{\alpha}=(\alpha_{1},\dots,\alpha_{r})$ are real and nonnegative
valued without any loss of generality. Such a form is also natural in
applications given that $\alpha_{p}$ usually captures the magnitude of
whatever quantity that is represented by the $p$ summand.

We will see this problem, whether in the form \eqref{approx1} or
\eqref{fund1}, has no solution in general. We will first observe a somewhat
unusual phenomenon in multilinear decomposition of $d$-partite functions where
$d\geq3$, namely, a sequence of rank-$r$ functions (each with an rank-$r$
multilinear decomposition) can converge to a limit that is not rank-$r$ (has
no rank-$r$ multilinear decomposition).

\begin{example}[Multilinear approximation of functions]\label{lack-ex} For linearly
independent $\varphi_{1},\psi_{1}:X_{1}\rightarrow\mathbb{C}$, $\varphi
_{2},\psi_{2}:X_{2}\rightarrow\mathbb{C}$, $\varphi_{3},\psi_{3}%
:X_{3}\rightarrow\mathbb{C}$, let $\hat{f}:X_{1}\times X_{2}\times
X_{3}\rightarrow\mathbb{C}$ be
\begin{multline*}
\hat{f}(\mathbf{x}_{1},\mathbf{x}_{2},\mathbf{x}_{3}):=\psi_{1}(\mathbf{x}%
_{1})\varphi_{2}(\mathbf{x}_{2})\varphi_{3}(\mathbf{x}_{3})\\
+\varphi_{1}(\mathbf{x}_{1})\psi_{2}(\mathbf{x}_{2})\varphi_{3}(\mathbf{x}%
_{3})+\varphi_{1}(\mathbf{x}_{1})\varphi_{2}(\mathbf{x}_{2})\psi
_{3}(\mathbf{x}_{3}).
\end{multline*}
For $n\in\mathbb{N}$, define {\footnotesize
\begin{multline*}
f_{n}(\mathbf{x}_{1},\mathbf{x}_{2},\mathbf{x}_{3}):=\\
n\left[  \varphi_{1}(\mathbf{x}_{1})+\frac{1}{n}\psi_{1}(\mathbf{x}%
_{1})\right]  \left[  \varphi_{2}(\mathbf{x}_{2})+\frac{1}{n}\psi
_{2}(\mathbf{x}_{2})\right]  \left[  \varphi_{3}(\mathbf{x}_{3})+\frac{1}%
{n}\psi_{3}(\mathbf{x}_{3})\right] \\
-n\varphi_{1}(\mathbf{x}_{1})\varphi_{2}(\mathbf{x}_{2})\varphi_{3}%
(\mathbf{x}_{3}).
\end{multline*}
}Then
\begin{multline*}
f_{n}(\mathbf{x}_{1},\mathbf{x}_{2},\mathbf{x}_{3}) - \hat{f}(\mathbf{x}_{1},\mathbf{x}_{2},\mathbf{x}_{3}) 
= \frac{1}{n}\bigl[ \psi_{1}(\mathbf{x}_{1}%
)\psi_{2}(\mathbf{x}_{2})\varphi_{3}(\mathbf{x}_{3})\\
+\psi_{1}(\mathbf{x}_{1})\varphi_{2}(\mathbf{x}_{2})\psi_{3}+\varphi
_{1}(\mathbf{x}_{1})\psi_{2}(\mathbf{x}_{2})\psi_{3}(\mathbf{x}_{3})\bigr]\\
+ \frac{1}{n^2} \psi_{1}(\mathbf{x}_{1})\psi_{2}(\mathbf{x}_{2})\psi_{3}(\mathbf{x}_{3}). 
\end{multline*}
Hence
\begin{equation}
\lVert\hat{f}-f_{n}\rVert=O\left(  \frac{1}{n}\right)  . \label{o1n}%
\end{equation}

\end{example}

\begin{lemma}
In Example \ref{lack-ex}, $\operatorname*{rank}(\hat{f})=3$ iff $\varphi
_{i},\psi_{i}$ are linearly independent, $i=1,2,3$. Furthermore, it is clear
that $\operatorname*{rank}(f_{n})\leq2$ and
\[
\lim_{n\rightarrow\infty}f_{n}=\hat{f}.
\]

\end{lemma}

Note that our fundamental approximation problem may be regarded as the
approximation problem%
\begin{equation}
\operatorname*{argmin} \{\lVert\hat{f}-f\rVert:\operatorname*{rank}(f)\leq r\},
\label{rankr}%
\end{equation}
followed by a decomposition problem%
\[
f=\sum_{p=1}^{r}\alpha_{p}\prod_{k=1}^{d}\varphi_{kp},
\]
which always exists for an $f$ with $\operatorname*{rank}(f)\leq r$. The
discussion above shows that there are target functions $\hat f$ for which%
\[
\operatorname*{argmin} \{\lVert\hat{f}-f\rVert:\operatorname*{rank}(f)\leq
r\}=\varnothing,
\]
and thus \eqref{fund1} or \eqref{rankr} does not need to have a solution in
general. This is such a crucial point that we are obliged to formally state it.

\begin{theorem}
\label{thm:nonexist}For $d\geq3$, the best approximation of a $d$-partite
function by a sum of $p$ products of $d$ separable functions does not exist in general.
\end{theorem}

\begin{proof}
Take the tripartite function $\hat{f}\in L^{2}(X_{1}\times X_{2}\times X_{3})$
in Example~\ref{lack-ex}. Suppose we seek a best rank-$2$ approximation, in
other words, we seek to solve the minimization problem%
\[
\operatorname*{argmin}_{\lVert g_{k}\rVert=\lVert h_{k}\rVert=1,\;\gamma
,\eta\geq0}\Vert\hat{f}-\gamma g_{1}\otimes g_{2}\otimes g_{3}-\eta
h_{1}\otimes h_{2}\otimes h_{3}\Vert.
\]
Now, the \textit{infimum},%
\[
\inf_{\lVert g_{k}\rVert=\lVert h_{k}\rVert=1,\;\gamma,\eta\geq0}\Vert\hat
{f}-\gamma g_{1}\otimes g_{2}\otimes g_{3}-\eta h_{1}\otimes h_{2}\otimes
h_{3}\Vert=0
\]
since we may choose $n\in\mathbb{N}$ sufficiently large,%
\[
g_{k}=\frac{\varphi_{k}+n^{-1}\psi_{k}}{\lVert\varphi_{k}+n^{-1}\psi_{k}%
\rVert},\quad h_{k}=\frac{\varphi_{k}}{\lVert\varphi_{k}\rVert},
\]
for $k=1,2,3$,%
\begin{align*}
\gamma &  =n\lVert\varphi_{1}+n^{-1}\psi_{1}\rVert\lVert\varphi_{2}+n^{-1}%
\psi_{2}\rVert\lVert\varphi_{3}+n^{-1}\psi_{3}\rVert,\\
\eta &  =n\lVert\varphi_{1}\rVert\lVert\varphi_{2}\rVert\lVert\varphi
_{3}\rVert,
\end{align*}
so as make $\Vert\hat{f}-\gamma g_{1}\otimes g_{2}\otimes g_{3}-\eta
h_{1}\otimes h_{2}\otimes h_{3}\Vert$ as small as we desired by virtue of
\eqref{o1n}. However there is no rank-$2$ function $\gamma g_{1}\otimes
g_{2}\otimes g_{3}+\eta h_{1}\otimes h_{2}\otimes h_{3}$ for which%
\[
\Vert\hat{f}-\gamma g_{1}\otimes g_{2}\otimes g_{3}-\eta h_{1}\otimes
h_{2}\otimes h_{3}\Vert=0.
\]
In other words, the zero infimum can never be attained.
\end{proof}

Our construction above is based on an earlier construction in
\cite{DesiL08:simax}. The first such example was given in
\cite{BiniLR80:siamjc}, which also contains the very first definition of border
rank. We will define it here for $d$-partite functions. When $X_{1}%
,\dots,X_{d}$ are finite sets, this reduces to the original definition in
\cite{BiniLR80:siamjc} for hypermatrices.

\begin{definition}
Let $f\in L^{2}(X_{1}\times\dots\times X_{d})$. The \textbf{border rank} of
$f$ is defined as%
\begin{multline*}
\overline{\operatorname*{rank}}(f)=\min\{r\in\mathbb{N}:\inf\lVert
f-g\rVert=0\\
\text{over all }g\text{ with }\operatorname*{rank}(g)=r\}.
\end{multline*}
We say $\overline{\operatorname*{rank}}(f) =\infty$ if such a finite $r$ does not exist.
\end{definition}

Clearly we would always have that%
\[
\overline{\operatorname*{rank}}(f)\leq\operatorname*{rank}(f).
\]
The discussions above show that strict inequality can occur. In fact, for the
$\hat{f}$ in Example~\ref{lack-ex}, $\overline{\operatorname*{rank}}(\hat
{f})=2$ while $\operatorname*{rank}(\hat{f})=3$.

We would like to mention here that this problem applies to operators too.
Optimal approximation of an operator by a sum of tensor/Kronecker products of
lower-dimensional operators, which arises in numerical operator calculus \cite{BM}, is in general an ill-posed problem whose existence cannot be guaranteed.

\begin{example}[Multilinear approximation of operators]\label{eg:oper}For linearly
independent operators $\Phi_{i},\Psi_{i}:V_{i}\rightarrow W_{i}$, $i=1,2,3$,
let $\widehat{T}:V_{1}\otimes V_{2}\otimes V_{3}\rightarrow W_{1}\otimes
W_{2}\otimes W_{3}$ be
\begin{equation}
\widehat{T}:=\Psi_{1}\otimes\Phi_{2}\otimes\Phi_{3}+\Phi_{1}\otimes\Psi
_{2}\otimes\Phi_{3}+\Phi_{1}\otimes\Phi_{2}\otimes\Psi_{3}. \label{oper}%
\end{equation}
If $\Phi_{i},\Psi_{i}$'s are all finite-dimensional and represented in
coordinates as matrices, then `$\otimes$' may be taken to be Kronecker product
of matrices. For $n\in\mathbb{N}$,
\begin{multline*}
T_{n}:=n\left[  \Phi_{1}+\frac{1}{n}\Psi_{1}\right]  \otimes\left[  \Phi
_{2}+\frac{1}{n}\Psi_{2}\right]  \otimes\left[  \Phi_{3}+\frac{1}{n}\Psi
_{3}\right] \\
-n\Phi_{1}\otimes\Phi_{2}\otimes\Phi_{3}.
\end{multline*}
Then
\[
\lim_{n\rightarrow\infty}T_{n}=\widehat{T}.
\]
An example of an operator that has the form in \eqref{oper} is the
$3m$-dimensional Laplacian $\Delta_{3m}$, which can be expressed in terms of
the $m$-dimensional Laplacian $\Delta_{m}$ as%
\[
\Delta_{3m}=\Delta_{m}\otimes I\otimes I+I\otimes\Delta_{m}\otimes I+I\otimes
I\otimes\Delta_{m}.
\]

\end{example}

There are several simple but artificial ways to alleviate the issue of
non-existent best approximant. Observe from the proof of
Theorem~\ref{thm:nonexist} that the coefficients in the approximant
$\gamma,\eta$ becomes unbounded in the limit. Likewise we see this happening
in Example~\ref{eg:oper}. In fact this must \textit{always} happen --- in the
event when a function or operator is approximated by a rank-$r$ function, i.e.%
\begin{equation}
\left\Vert \hat{f}-\sum_{p=1}^{r}\alpha_{p}\prod_{k=1}^{d}\varphi
_{kp}\right\Vert \quad\text{or}\quad\left\Vert \widehat{T}-\sum_{p=1}%
^{r}\alpha_{p}\bigotimes_{k=1}^{d}\Phi_{kp}\right\Vert , \label{eq:loss}%
\end{equation}
and if a best approximation does not exist, then the $r$ coefficients
$\alpha_{1},\dots,\alpha_{r}$ must \textit{all} diverge in magnitude to
$+\infty$ as the approximant converges to the infimum of the norm loss
function in \eqref{eq:loss}. This result was first established in
\cite[Proposition~4.9]{DesiL08:simax}.

So a simple but artificial way to prevent the nonexistence issue is to simply
limit the sizes of the coefficients $\alpha_{1},\dots,\alpha_{r}$ in the
approximant. One way to achieve this is regularization \cite{Paat97:cils,
LimC09:jchemo} --- adding a regularization term to our objective function in
\eqref{fund1} to penalize large coefficients. A common choice is Tychonoff
regularization, which uses a sum-of-squares regularization term:%
\begin{equation}
\operatorname*{argmin}_{\boldsymbol{\alpha}\in[0,\infty)^{r},\;\lVert
\varphi_{kp}\rVert=1}\left\Vert \hat{f}-\sum_{p=1}^{r}\alpha_{p}\prod
_{k=1}^{d}\varphi_{kp}\right\Vert +\lambda\sum_{p=1}^{r}\lvert\alpha_{p}%
\rvert^{2}. \label{cond3}%
\end{equation}
Here, $\lambda$ is an arbitrarily chosen regularization parameter. It can be
seen that this is equivalent to constraining the sizes $\alpha_{1}%
,\dots,\alpha_{r}$ to $\sum_{p=1}^{r}\lvert\alpha_{p}\rvert^{2}=\rho$, with
$\rho$ being determined a posteriori from $\lambda$. The main drawback of such
constraints is that $\rho$ and $\lambda$ are arbitrary, and that they
generally have no physical meaning.

More generally, one may alleviate the nonexistence issue by restricting the
optimization problem \eqref{rankr} to a compact subset of its non-compact
feasible set%
\[
\{f\in L^{2}(X_{1}\times\dots\times X_{d}):\operatorname*{rank}(f)\leq r\}.
\]
Limiting the sizes of $\alpha_{1},\dots,\alpha_{r}$ is a special case but
there are several other simple (but also artificial) strategies. In
\cite{Como92:elsevier}, the factors $\varphi_{k1},\dots,\varphi_{kp}$ are
required to be orthogonal \textit{for all} $k\in\{1,\dots,d\}$, i.e.%
\begin{equation}
\langle\varphi_{kp},\varphi_{kq}\rangle_{k}=\delta_{pq},\quad p,q=1,\dots
,r,\quad k=1,\dots,d. \label{cond1}%
\end{equation}
This remedy is acceptable only in very restrictive conditions. In fact a
necessary condition for this to work is that%
\[
r\leq\min_{k=1,\dots,d}\dim L^{2}(X_{k}).
\]
It is also trivial to see that imposing orthogonality between the separable
factors removes this problem%
\begin{equation}
\langle\varphi_{1p}\otimes\dots\otimes\varphi_{dp},\varphi_{1q}\otimes
\dots\otimes\varphi_{dq}\rangle=\delta_{pq},\quad p,q=1,\dots,r. \label{cond2}%
\end{equation}
This constraint is slightly less restrictive --- by \eqref{ip-eq}, it is
equivalent to requiring \eqref{cond1} \textit{for some} $k\in\{1,\dots,d\}$.
Both \eqref{cond1} and \eqref{cond2} are nonetheless so restrictive as to
exclude the most useful circumstances for the model \eqref{CP-eq}, which
usually involves factors that have no reason to be orthogonal, as we will see
in Section~\ref{applications-sec}. In fact, Kruskal's uniqueness condition is
such a potent tool precisely because it does not require orthogonality.

The conditions \eqref{cond1}, \eqref{cond2}, and \eqref{cond3} all limit the
feasible sets for the original approximation \eqref{fund1} to a much smaller
compact subset of the original feasible set. This is not the case for
nonnegative constraints. In \cite{LimC09:jchemo} it was shown that the
following best rank-$r$ approximation problem for a nonnegative-valued
$\hat{f}$ and where the coefficients $\alpha_{p}$ and factors $\varphi_{kp}$
of the approximants are also nonnegative valued, i.e.%
\[
\operatorname*{argmin}_{\boldsymbol{\alpha}\in[0,\infty)^{r}%
,\;\lVert\varphi_{kp}\rVert=1,\;\varphi_{kp}\geq0}\left\Vert \hat{f}%
-\sum_{p=1}^{r}\alpha_{p}\prod_{k=1}^{d}\varphi_{kp}\right\Vert ,
\]
always has a solution. The feasible set in this case is non-compact and has
nonempty interior within the feasible set of our original problem
\eqref{fund1}. The nonnegativity constraints are natural in some applications,
such as the fluorescence spectroscopy one described in Section~\ref{sec:fluo},
where $\varphi_{kp}$ represent intensities and concentrations, and are
therefore nonnegative valued.

There are two major problems with imposing artificial constraints simply to
force a solution: How do we know a priori that the solution that we seek would
meet those constraints? But more importantly, perhaps the model is ill-posed
and a solution indeed should not exist? To illustrate the case in point with a
more commonplace example, suppose we want to find a maximum likelihood
estimator $X\in\mathbb{R}^{n\times n}$ for the covariance $\Sigma$ of
independent samples $\mathbf{y}_{1},\dots,\mathbf{y}_{m}\sim N(0,\Sigma)$.
This would lead us to a semi-definite programming problem%
\begin{equation}
\operatorname*{argmin}_{X\succ0}\operatorname*{tr}(X^{-1}Y)-\log\det(X)
\label{eq:mle}%
\end{equation}
where $Y=\frac{1}{m}\sum_{i=1}^{m}\mathbf{y}_{i}\mathbf{y}_{i}^{\mathsf{T}}$.
However the problem will not have a solution when the number of samples is
smaller than the dimension, i.e., $m<n$, as the infimum of the loss function
in \eqref{eq:mle} cannot be attained by any $X$ in the feasible set. This is
an indication that we should seek more samples (so that we could get $m\geq
n$, which will guarantee the attainment of the infimum) or use a different
model (e.g., determine if $X^{-1}$ might have some a priori zero
entries due to statistical independence of the variables). It is usually
unwise to impose artificial constraints on the covariance matrix $X$ just so
that the loss function in \eqref{eq:mle} would attain an infimum on a smaller
feasible set --- the thereby obtained `solution' may bear no relation to the
true solution that we want.

Our goal in Section~\ref{sec:exist} is to define a type of physically
meaningful constraints via the notion of \textit{coherence}. It ensures the
existence of a unique minimum, but not via an artificial limitation of the
optimization problem to a convenient subset of the feasible set. In the
applications we discuss in Section~\ref{applications-sec}, we will see that it
is natural to expect existence of a solution when coherence is small enough,
but not otherwise. So when our model is ill-posed or ill-conditioned, we are
warned by the size of the coherence and could seek other remedies (collect
more measurements, use a different model, etc) instead of forcing a `solution'
that bears no relation to reality. But before we get to that we will examine
why, unlike in compressed sensing and matrix completion, the approximation of rank by a ratio of nuclear and spectral norms could not be expected to work here.

\section{Nuclear and spectral norms}\label{sec:nuclear}

We introduce the notion of nuclear and spectral norms for multipartite
functions. Our main purpose is to see if they could be used to alleviate the
problem discussed in Section~\ref{illposed-sec}, namely, that a $d$-partite
function may not have a best approximation by a sum of $r$ separable functions.

The definition of nuclear norm follows naturally from the definition of rank
in Section~\ref{problem-sec}.
\begin{definition}
\label{def:nuclear}We define the \textbf{nuclear norm} (or \textbf{Schatten $1$-norm})
of $f\in L^{2}(X_{1}\times\dots\times X_{d})$ as%
\begin{multline}
\lVert f\rVert_{\ast}:=\inf\Biggl\{ \Biggl[ \sum_{p=1}^{\infty}\lambda
_{p}\Biggr]: f=\sum_{p=1}^{\infty}\lambda_{p}\varphi_{1p}%
\otimes\dots\otimes\varphi_{dp},\\
\lVert\varphi_{kp}\rVert=1,\;\lambda_{p}\geq\lambda_{p+1}%
>0\Biggr\}. \label{eq:nuclear}
\end{multline}
\end{definition}
Note that for rank-$1$ functions, we always have that
\begin{equation}\label{eq:multnorm}
 \lVert \varphi_{1}\otimes\dots\otimes\varphi_{d}\rVert_{\ast} 
=  \lVert \varphi_1 \rVert  \cdots  \lVert \varphi_d \rVert .
\end{equation}
A finite rank function always has finite nuclear norm but in general a function in $ L^{2}(X_{1}\times\dots\times X_{d})$
need not have finite nuclear norm.

The definition of the spectral norm of a multipartite function is motivated by the fact the usual
spectral norm of a matrix $A$ equals the maximal absolute value of its inner product $\operatorname{tr}(AX)$ with rank-$1$ unit-norm matrices $X = \mathbf{u}\mathbf{v}^{\ast}$, $\lVert \mathbf{u} \rVert_2 = \lVert \mathbf{v} \rVert_2 = 1$.
\begin{definition}
\label{def:spectral}We define the \textbf{spectral norm} 
of $f\in L^{2}(X_{1}\times\dots\times X_{d})$ as%
\begin{multline}
\lVert f\rVert_{\sigma}:=\sup \bigl\{ \lvert \langle f, \varphi_1 \otimes \dots \otimes \varphi_d \rangle \rvert : \\
\lVert \varphi_1 \rVert = \dots = \lVert \varphi_d \rVert =1 \bigr\}. \label{eq:spectral}
\end{multline}
\end{definition}
Here we write $\lVert \, \cdot \, \rVert$ for the $L^2$-norms in $L^2(X_k)$, $k =1,\dots, d$. Alternatively, we may also use $\operatorname{Re} \langle f, \varphi_1 \otimes \dots \otimes \varphi_d \rangle$ in place of $\lvert \langle f, \varphi_1 \otimes \dots \otimes \varphi_d \rangle \rvert$ in \eqref{eq:spectral}, which does not change its value. Note that a function in $ L^{2}(X_{1}\times\dots\times X_{d})$ always has finite spectral norm.

The fact that \eqref{eq:nuclear} and \eqref{eq:spectral} define norms on $L^{2}(X_{1}\times\dots\times X_{d})$ follows from the
standard Minkowski gauge argument \cite{DF}. Suppose
$X_{1},\dots,X_{d}$ are finite sets of cardinalities $n_{1},\dots,n_{d}%
\in\mathbb{N}$. The nuclear and spectral norms for the unipartite case ($d=1$) are 
the $\ell^{1}$- and  $\ell^{\infty}$-norms for vectors in $\mathbb{C}^{n_{1}}=L^{2}(X_{1})$.
The nuclear and spectral norms for the bipartite case ($d=2$) agrees with the usual
nuclear and spectral norms for matrices in $\mathbb{C}^{n_{1}\times n_{2}}=L^{2}%
(X_{1}\times X_{2})$. For general $d\geq3$,
Definition~\ref{def:nuclear} yields a notion of nuclear norm\footnote{The notion of a nuclear norm for tensors was originally introduced in Section~3 of our 2010 article (cf.\ \url{http://arxiv.org/abs/1002.4935v1}). However, it was ultimately removed in the published version \cite{LimC10:cras} because of the page limit of the journal.} for hypermatrices in $\mathbb{C}^{n_{1}\times\dots\times
n_{d}}=L^{2}(X_{1}\times\dots\times X_{d})$, while
Definition~\ref{def:spectral} agrees with the usual notion of spectral norm for hypermatrices \cite{L}.

\begin{example}[Nuclear and spectral norms for $3$-tensors]Let $T\in\mathbb{C}^{n_{1}\times n_{2}\times
n_{3}}$. Then by Definition~\ref{def:nuclear}, we have%
\[
\lVert T\rVert_{\ast}=\inf\left\{  \sum_{p=1}^{r}\lambda_{p}:T=\sum
_{p=1}^{r}\lambda_{p}\mathbf{u}_{p}\otimes\mathbf{v}_{p}\otimes\mathbf{w}%
_{p}\right\}  ,
\]
where the infimum is taken over all linear combinations of complex vectors of
unit $2$-norm $\mathbf{u}_{p}\in\mathbb{C}^{n_{1}}$, $\mathbf{v}_{p}%
\in\mathbb{C}^{n_{2}}$, $\mathbf{w}_{p}\in\mathbb{C}^{n_{3}}$, with real
positive coefficientss $\lambda_{p}\in[0,\infty)$, and $p=1,\dots,r$, with
$r\in\mathbb{N}$.

We shall write $\lVert \, \cdot \, \rVert =\lVert \, \cdot \, \rVert_{2}$. The spectral norm of $T$ is
\begin{align*}
\lVert T\rVert_{\sigma} &=\sup\left\{ \lvert \langle T, \mathbf{u}\otimes\mathbf{v}\otimes\mathbf{w} \rangle \rvert :
\lVert\mathbf{u}\rVert =\lVert\mathbf{v}\rVert=\lVert\mathbf{w}\rVert=1\right\}\\
&= \sup_{\mathbf{x},\mathbf{y},\mathbf{z}\ne\mathbf{0}}\frac{\lvert T(\mathbf{x}, \mathbf{y}, \mathbf{z}) \rvert}{\lVert\mathbf{x}\rVert\lVert\mathbf{y}\rVert\lVert\mathbf{z}\rVert}= \lVert T\rVert_{2,2,2}.
\end{align*}
We have regarded $T$ as a trilinear functional defined by $ T(\mathbf{x}, \mathbf{y}, \mathbf{z}) = \sum_{i,j,k=1}^{n_1, n_2, n_3} t_{ijk}x_i y_j z_k$ and $\lVert T\rVert_{2,2,2}$ is its induced norm as defined in \cite{L,HLA}. Again, these clearly extend to any $d$-tensors.  We will see in Lemma~\ref{lem:dual} that the nuclear and spectral norms for tensors are dual to each other.
\end{example}

Note that we have used the term \textit{tensors}, as opposed to hypermatrices, in
the above example. In fact, Definition~\ref{def:nuclear}
defines nuclear norms for the tensors, not just their coordinate
representations as hypermatrices (see our discussion after
Example~\ref{eg:hitch}), because of the following invariant properties.

\begin{lemma}
The nuclear and spectral norms for $\mathbb{C}^{n_{1}\times\dots\times n_{d}}$ are unitarily
invariant, i.e., invariant under the natural action of $\operatorname{U}_{n_1}(\mathbb{C})\times
\dots\times \operatorname{U}_{n_d}(\mathbb{C})$ where $\operatorname{U}_{n}(\mathbb{C})$ denotes the group of unitary matrices in
$\mathbb{C}^{n\times n}$.
\end{lemma}

\begin{proof}
To avoid the clutter of indices, we will assume that $d=3$. It is easy,
although notationally cumbersome, to extend this to general $d\geq3$. Let
$(U,V,W)\in \operatorname{U}_{n_1}(\mathbb{C})\times \operatorname{U}_{n_2}(\mathbb{C})\times \operatorname{U}_{n_3}(\mathbb{C})$ and $T\in\mathbb{C}%
^{n_{1}\times n_{2}\times n_{3}}$. The natural action, given in coordinates by%
\[
(U,V,W)\cdot T=\left[  \sum_{i,j,k=1}^{n_{1},n_{2},n_{3}}u_{ai}v_{bj}%
w_{ck}t_{ijk}\right]  _{a,b,c=1}^{n_{1},n_{2},n_{3}},
\]
has the property that if $T$ has a multilinear decomposition of the form%
\[
T=\sum_{p=1}^{r}\lambda_{p}\mathbf{x}_{p}\otimes\mathbf{y}_{p}\otimes
\mathbf{z}_{p},
\]
then%
\begin{equation}
(U,V,W)\cdot T=\sum_{p=1}^{r}\lambda_{p}(U\mathbf{x}_{p})\otimes
(V\mathbf{y}_{p})\otimes(W\mathbf{z}_{p}). \label{eq1}%
\end{equation}
\eqref{eq1} is obvious when $r=1$ and for general $r$ follows from the
linearity of the action, i.e., $(U,V,W)\cdot(S+T)=(U,V,W)\cdot S+(U,V,W)\cdot
T$. We also need the simple fact that $\operatorname{U}_{n_1}(\mathbb{C})\times \operatorname{U}_{n_2}(\mathbb{C})\times \operatorname{U}_{n_3}(\mathbb{C})$ acts transitively on unit-norm rank-$1$ tensors, i.e., take any
$\mathbf{x}\in\mathbb{C}^{n_{1}}$, $\mathbf{y}\in\mathbb{C}^{n_{2}}$,
$\mathbf{z}\in\mathbb{C}^{n_{3}}$ of unit norm, then every other unit-norm
rank-$1$ tensor may be written as $U\mathbf{x}\otimes V\mathbf{y}\otimes
W\mathbf{z}$ for some $(U,V,W)\in \operatorname{U}_{n_1}(\mathbb{C})\times \operatorname{U}_{n_2}(\mathbb{C})\times \operatorname{U}_{n_3}(\mathbb{C})$.
With these, it follows immediately from Definition~\ref{def:nuclear} that
nuclear norms satisfy%
\[
\lVert(U,V,W)\cdot T\rVert_{\ast}=\lVert T\rVert_{\ast}.
\]
One may similarly show that the spectral norm is also unitarily invariant or deduce the fact from Lemma~\ref{lem:dual} below.
\end{proof}

Recall that on a Hilbert space $\mathbb{H}$ with inner product $\langle \cdot \, , \cdot \rangle$ the \textit{dual norm} of a given norm $\lVert \, \cdot \, \rVert$ is defined as
\[
\lVert f \rVert^\vee := \sup \bigl\{ \lvert \langle f, g \rangle \rvert : \lVert g \rVert \le 1 \}.
\]
If $\lVert \, \cdot \, \rVert$ is the norm induced by the inner product, then $\lVert \,\cdot\, \rVert^\vee = \lVert \,\cdot\, \rVert$; but in general they are different. Nevertheless one always have that $(\lVert \,\cdot\, \rVert^\vee)^\vee =\lVert \,\cdot\, \rVert$ and 
\begin{equation}\label{eq:dual0}
\lvert \langle f, g \rangle \rvert \le \lVert f \rVert^\vee \lVert g \rVert
\end{equation}
for any $f,g \in \mathbb{H}$.

Since the $\ell^1$- and $\ell^\infty$-norms on $\mathbb{C}^n$ are dual, as are the nuclear and spectral norms on $\mathbb{C}^{n_1 \times n_2}$, one may wonder if it is true in general that the nuclear and spectral norms are dual to each other. This is in fact almost a tautology when $X_1,\dots,X_d$ are finite.
\begin{lemma}\label{lem:dual}
Let $X_1,\dots,X_d$ be finite sets.
Then nuclear and spectral norms on $L^2 (X_1 \times \dots \times X_d)$ satisfy
\begin{equation}\label{eq:dual1}
\lvert \langle f, g \rangle \rvert \le \lVert f \rVert_{\sigma } \lVert g \rVert_{\ast }
\end{equation}
and in fact
\begin{equation}\label{eq:dual2}
\lVert f \rVert_{\ast}^\vee  = \lVert f \rVert_{\sigma}.
\end{equation}
\end{lemma}
\begin{proof}
We first need to establish \eqref{eq:dual1} without invoking \eqref{eq:dual0}. Since $X_1,\dots,X_d$ are finite, any $g \in L^2 (X_1 \times \dots \times X_d)$ is of finite rank.
Take any multilinear decomposition
\[
g=\sum_{p=1}^{r}\lambda_{p}\varphi_{1p}\otimes\dots\otimes\varphi_{dp}
\]
where $r \in \mathbb{N}$ is arbitrary. Now
\begin{align*}
\lvert \langle f, g \rangle \rvert  &\le \sum_{p=1}^r \lvert \lambda_p \rvert \lvert \langle f, \varphi_{1p}\otimes\dots\otimes\varphi_{dp} \rangle \rvert \\
&\le \lVert f \rVert_{\sigma} \sum_{p=1}^r \lvert \lambda_p \rvert
\end{align*}
by definition of spectral norm. Taking infimum over all finite-rank decompositions, we arrive at \eqref{eq:dual1} by definition of nuclear norm. Hence
\begin{align*}
\lVert f \rVert_{\ast}^\vee  &=   \sup \bigl\{ \lvert \langle f, g \rangle \rvert : \lVert g \rVert_{\ast} \le 1 \} \\
&\le \sup \bigl\{ \lVert f \rVert_{\sigma } \lVert g \rVert_{\ast } : \lVert g \rVert_{\ast} \le 1 \} =\lVert f \rVert_{\sigma}.
\end{align*}
On the other hand, using \eqref{eq:dual0} for $ \lVert \,\cdot\, \rVert_{\ast}$ and  $\lVert \,\cdot\, \rVert_{\ast}^\vee$, we get
\begin{align*}
\lVert f \rVert_{\sigma}  &=   \sup \bigl\{ \lvert \langle f, \varphi_{1}\otimes\dots\otimes\varphi_{d}\rangle \rvert : \lVert \varphi_k \rVert  =1 \bigr\}\\
&\le \sup \bigl\{ \lVert f \rVert_{\ast}^\vee \lVert \varphi_{1}\otimes\dots\otimes\varphi_{d}\rVert_{\ast} : \lVert \varphi_k \rVert =1 \bigr\} = \lVert f \rVert_{\ast}^\vee.
\end{align*}
where the last equality follows from \eqref{eq:multnorm}.
\end{proof}
When $X_1,\dots, X_d$ are only required to be $\sigma$-finite measurable spaces, we may use a limiting argument to show that \eqref{eq:dual1} still holds for $f$ of finite spectral norm and $g$ of finite nuclear norm; a proper generalization of \eqref{eq:dual2} is more subtle and we will leave this to future work since we do not require it in this article.

It is known \cite{DonoE03:pnas} that the $\ell^{1}$-norm is the \textit{largest} convex
underestimator\footnote{Also called the \textit{greatest convex minorant}, in
this case also equivalent to the \textit{Legendre-Frenchel biconjugate} or
\textit{convex biconjugate}.} of the ``$\ell^0$-norm'' on the $\ell^{\infty}$-norm
unit ball \cite{HL} and that the nuclear norm is the \textit{largest} convex
underestimator of rank on spectral norm unit ball \cite{FHB}. In particular,
\[
\lVert \mathbf{x} \rVert_1 \le \operatorname{nnz}(\mathbf{x}) \lVert \mathbf{x} \rVert_\infty
\]
for all $\mathbf{x} \in \mathbb{C}^{n}$ while
\[
\lVert X \rVert_* \le \operatorname{rank}(X) \lVert X \rVert_\sigma
\]
for all $X \in \mathbb{C}^{m \times n}$. The quantity $\operatorname{nnz}(\mathbf{x}) := \#\{ i : x_i\ne 0 \}$ is often called ``$\ell^0$-norm'' even though it is not a norm (and neither a seminorm nor a quasinorm nor a pseudonorm). 

We had suspected that the following generalization might perhaps be true, namely, rank, nuclear, and spectral norms as defined in \eqref{rank-eq}, \eqref{eq:nuclear}, and \eqref{eq:spectral} would also satisfy the same inequality:
\begin{equation}
\lVert f\rVert_{\ast}\leq\operatorname*{rank}(f)\times\lVert f\rVert
_{\sigma}. \label{eq:holder}%
\end{equation}
If true, this would immediately imply the same for border rank
\[
\lVert f\rVert_{\ast}\leq\overline{\operatorname*{rank}}(f)\times\lVert
f\rVert_{\sigma}
\]
by a limiting argument. Furthermore, \eqref{eq:holder} would provide a simple way to remedy the nonexistence problem highlighted in Theorem~\ref{thm:nonexist}: One may use the ratio $\lVert f\rVert
_{\ast}/\lVert f\rVert_{\sigma}$ as a `proxy' in place of
$\operatorname*{rank}(f)$ and replace the condition
$\operatorname*{rank}(f)\leq r$ by the weaker condition $\lVert f\rVert_{\ast}\leq
r\lVert f\rVert_{\sigma}$. The discussion in
Section~\ref{illposed-sec} shows that there are $\hat{f}$ for which%
\[
\operatorname*{argmin}\{\lVert\hat{f}-f\rVert:\operatorname*{rank}(f)\leq
r\}=\varnothing,
\]
which really results from the fact that%
\[
\{f\in L^{2}(X_{1}\times\dots\times X_{d}):\operatorname*{rank}(f)\leq r\}
\]
is not a closed set. But
\begin{equation}
\{f\in L^{2}(X_{1}\times\dots\times X_{d}):\lVert f\rVert_{\ast}\leq r\lVert
f\rVert_{\sigma}\} \label{eq:closedcond}
\end{equation}
is always closed (by the continuity of norms) and so for
any $r\in\mathbb{N}$, the optimization problem%
\[
\operatorname*{argmin}\{\lVert\hat{f}-f\rVert:\lVert f\rVert_{\ast}\leq
r\lVert f\rVert_{\sigma}\}
\]
always has a solution.

Unfortunately, \eqref{eq:holder} is not true when $d > 2$. The following example shows that nuclear norm is not an underestimator of rank on the spectral norm unit ball, and can in fact be arbitrarily larger than rank on the spectral norm unit ball.

\begin{example}[Matrix multiplication]
Let $T_n \in \mathbb{C}^{n^2 \times n^2 \times n^2}$ be the matrix multiplication tensor (cf.\ Applications~2 and 3 in \cite[Section 15.3]{HLA}). The well-known result of Strassen \cite{Strassen} implies that
\[
\operatorname{rank}(T_{n}) \le c n^{\log_2 7}
\]
for some $c > 0$ and for $n$ sufficiently large. On the other hand, Derksen \cite{Derksen2} has recently established the exact values for the nuclear and spectral norm of $T_n$:
\[
\lVert T_{n} \rVert_{*} = n^{3}, \quad \lVert T_{n} \rVert_{\sigma} =1
\]
for all $n \in \mathbb{N}$. It then follows that
\[
\lim_{n\to \infty} \frac{\lVert T_{n} \rVert_{*} }{\operatorname{rank}(T_{n}) } = \infty.
\]
\end{example}

Fortunately, we do not need to rely on \eqref{eq:closedcond} for the applications consider in this article. Instead another workaround that uses the notion of coherence, discussed in the next section, is more
naturally applicable in our situtations.

\section{Coherence\label{angular-sec}}

We will show in this section that a simple measure of angular constraint
called coherence, or rather, the closely related notion of \textit{relative
incoherence}, allows us to alleviate two problems simultaneously: the
computational intractability of checking for uniqueness discussed in
Section~\ref{sec:unique} and the non-existence of a best approximant in
Section~\ref{illposed-sec}.

\begin{definition}\label{coherence-def}
Let $\mathbb{H}$ be a Hilbert space provided with scalar
product $\langle\cdot,\cdot\rangle$, and let $\Phi\subseteq\mathbb{H}$ be a
set of elements of unit norm in $\mathbb{H}$. The \textbf{coherence} of $\Phi$
is defined as%
\[
\mu(\Phi)=\sup_{\varphi\neq\psi}\, \lvert\langle\varphi,\psi\rangle\rvert
\]
where the supremum is taken over all distinct pairs $\varphi,\psi\in
\Phi$. If $\Phi=\{\varphi_{1},\dots,\varphi_{r}\}$ is finite, we also
write $\mu(\varphi_{1},\dots,\varphi_{r}):=\max_{p\neq q}\lvert\langle
\varphi_{p},\varphi_{q}\rangle\rvert$.
\end{definition}

We adopt the convention that whenever we write $\mu(\Phi)$ (resp.$\ \mu
(\varphi_{1},\dots,\varphi_{r})$) as in Definition~\ref{coherence-def}, it is
implicitly implied that all elements of $\Phi$ (resp.$\ \varphi_{1}%
,\dots,\varphi_{r}$) are of unit norm.

The notion of coherence has received different names in the literature: mutual
incoherence of two dictionaries \cite{DonoE03:pnas}, mutual coherence of two
dictionaries \cite{CandR07:ip}, the coherence of a subspace projection
\cite{CandT10:it}, etc. The version here follows that of \cite{GribN03:it}.
Usually, dictionaries are finite or countable, but we have here a continuum of
atoms. Clearly, $0\leq\mu(\Phi)\leq1$, and $\mu(\Phi)=0$ iff $\varphi
_{1},\dots,\varphi_{r}$ are orthonormal. Also, $\mu(\Phi)=1$ iff $\Phi$
contains at least a pair of collinear elements, i.e., $\varphi_{p}%
=\lambda\varphi_{q}$ for some $p\neq q$, $\lambda\neq0$.

We find it useful to introduce a closely related notion that we call relative
incoherence. It allows us to formulate some of our results slightly more elegantly.

\begin{definition}
\label{def:relcohe} Let $\Phi\subseteq\mathbb{H}$ be a set of elements of unit
norm. The \textbf{relative incoherence} of $\Phi$ is defined as
\[
\omega(\Phi)=\frac{1-\mu(\Phi)}{\mu(\Phi)}.
\]
For a finite set of unit vectors $\Phi=\{\varphi_{1},\dots,\varphi_{r}\}$, we
will also write $\omega(\varphi_{1},\dots,\varphi_{r})$ occasionally.
\end{definition}

It follows from our observation about coherence that $0\leq\omega(\Phi
)\leq\infty$, $\omega(\Phi)=\infty$ iff $\varphi_{1},\dots,\varphi_{r}$ are
orthonormal, and $\omega(\Phi)=0$ iff $\Phi$ contains at least a pair of
collinear elements.

In the next few subsections, we will see respectively how coherence can inform
us about the existence (Section~\ref{sec:exist}), uniqueness
(Section~\ref{sec:uniquecoh}), as well as both existence and uniqueness
(Section~\ref{sec:bothcoh}) of a solution to the best rank-$r$ multilinear
approximation problem \eqref{fund1}. We will also see how it can be used for
establishing exact recoverability (Section~\ref{sec:exact}) and approximation
bounds (Section~\ref{sec:greedy}) in greedy algorithms.

\subsection{Existence via coherence\label{sec:exist}}

The goal is to prevent the phenomenon we observed in Example~\ref{lack-ex} to
occur, by imposing natural and weak constraints; we do not want to reduce the
search to a compact set. It is clear that the objective is not coercive, which
explains why the minimum may not exist. But with an additional condition on
the \textit{coherence}, we shall be able to prove existence thanks to coercivity.

The following shows that a solution to the bounded coherence best rank-$r$
approximation problem always exists:

\begin{theorem}
\label{existence-prop}Let $f\in L^{2}(X_{1}\times\dots\times X_{d})$ be a
$d$-partite function. If
\begin{equation}
\prod_{k=1}^{d}(1+\omega_{k})>r-1 \label{eq2}%
\end{equation}
or equivalently if
\begin{equation}
\prod_{k=1}^{d}\mu_{k}<\frac{1}{r-1}, \label{eq3}%
\end{equation}
where $\mu_k$ denotes the coherence as in Definition~\ref{coherence-def}, then
\begin{align}
\eta &  =\inf\biggl\{\biggl\lVert f-\sum_{p=1}^{r}\lambda_{p}\varphi
_{1p}\otimes\dots\otimes\varphi_{dp}\biggr\rVert:\nonumber\\
&  \qquad\qquad\qquad\boldsymbol{\lambda}\in\mathbb{C}^{r},\;\mu(\varphi
_{k1},\dots,\varphi_{kr})\leq\mu_{k}\biggr\}\label{bdapprox-eq}\\
&  =\inf\biggl\{\biggl\lVert f-\sum_{p=1}^{r}\lambda_{p}\varphi_{1p}%
\otimes\dots\otimes\varphi_{dp}\biggr\rVert:\nonumber\\
&  \qquad\qquad\qquad\boldsymbol{\lambda}\in\mathbb{C}^{r},\;\omega
(\varphi_{k1},\dots,\varphi_{kr})\geq\omega_{k}\biggr\}\nonumber
\end{align}
is attained. Here, $\Vert\cdot\Vert$ denotes the $L^{2}$-norm on $L^{2}%
(X_{1}\times\dots\times X_{d})$ and $\boldsymbol{\lambda}=(\lambda_{1}%
,\dots,\lambda_{r})$. If desired, we may assume that $\boldsymbol{\lambda}%
\in\mathbb{R}^{r}$ and $\lambda_{1}\geq\dots\geq\lambda_{r}>0$ by
Theorem~\ref{thm:svd}.
\end{theorem}

\begin{proof}
The equivalence between \eqref{eq2} and \eqref{eq3} follows from
Definition~\ref{def:relcohe}. We show that if either of these conditions are
met, then the loss function is coercive. We have the following inequalities
\begin{flalign*}
&\left\Vert \sum_{p=1}^{r}\lambda_{p}\varphi_{1p}\otimes\dots\otimes
\varphi_{dp}\right\Vert ^{2}  =\sum_{p,q=1}^{r}\lambda_{p}\bar{\lambda}%
_{q}\prod_{k=1}^{d}\langle\varphi_{kp},\varphi_{kq}\rangle\\
&  \geq\sum_{p=1}^{r}\lambda_{p}\bar{\lambda}_{p}\prod_{k=1}^{d}\Vert
\varphi_{kp}\Vert^{2}-\sum_{p\neq q}^{r}\left\vert \lambda_{p}\bar{\lambda}_{q}\prod
_{k=1}^{d}\langle\varphi_{kp},\varphi_{kq}\rangle\right\vert \\
&  \geq\sum_{p=1}^{r}|\lambda_{p}|^{2}-\prod_{k=1}^{d}\mu_{k}\sum_{p\neq
q}\lvert\lambda_{p}\bar{\lambda}_{q}\rvert \geq\Vert\boldsymbol{\lambda}\Vert_{2}^{2}-(r-1)\Vert\boldsymbol{\lambda
}\Vert_{2}^{2}\prod_{k=1}^{d}\mu_{k}%
\end{flalign*}
where the last inequality follows from%
\[
\sum_{p\neq q}\lvert\lambda_{p}\bar{\lambda}_{q}\rvert \leq  (r-1)\Vert\boldsymbol{\lambda
}\Vert_{2}^{2},
\]
which is true because $\sum_{p\neq q} (\lvert\lambda_{p}\rvert-\lvert\lambda_{q}\rvert )^2\geq0$. 
This yields%
\begin{equation}
\left\Vert \sum_{p=1}^{r}\lambda_{p}\varphi_{1p}\otimes\dots\otimes
\varphi_{dp}\right\Vert ^{2}\geq\left[  1-(r-1)\prod_{k=1}^{d}\mu_{k}\right]
\Vert\boldsymbol{\lambda}\Vert_{2}^{2} \label{coercivity-eq}%
\end{equation}
Since by assumption $(r-1)\prod_{k=1}^{d}\mu_{k}<1$, it is clear that the left
hand side of \eqref{coercivity-eq} tends to infinity as $\Vert
\boldsymbol{\lambda}\Vert_{2}\rightarrow\infty$. And because $f$ is fixed,
$\left\Vert f-\sum_{p=1}^{r}\lambda_{p}\varphi_{1p}\otimes\dots\otimes
\varphi_{dp}\right\Vert $ also tends to infinity as $\Vert\boldsymbol{\lambda
}\Vert_{2}\rightarrow\infty$. This proves coercivity of the loss function and
hence the existential statement.
\end{proof}

The condition \eqref{eq2} or, equivalently, \eqref{eq3}, in
Theorem~\ref{existence-prop} is sharp in an appropriate sense.
Theorem~\ref{existence-prop} shows that the condition \eqref{eq3} is
sufficient in the sense that it guarantees a best rank-$r$ approximation when
the condition is met. We show that it is also necessary in the sense that if
\eqref{eq3} does not hold, then there are examples where a best rank-$r$
approximation fails to exist.

In fact, let $\hat{f}$ be as in Example~\ref{lack-ex}. As demonstrated in the
proof of Theorem~\ref{thm:nonexist}, the infimum for the case $d=3$ and $r=2$,%
\[
\inf_{\lVert g_{k}\rVert=\lVert h_{k}\rVert=1,\;\lambda,\mu\geq0}\Vert\hat
{f}-\lambda g_{1}\otimes g_{2}\otimes g_{3}-\mu h_{1}\otimes h_{2}\otimes
h_{3}\Vert
\]
is not attained. Since%
\[
g_{k}=\frac{\varphi_{k}+n^{-1}\psi_{k}}{\lVert\varphi_{k}+n^{-1}\psi_{k}%
\rVert},\quad h_{k}=\frac{\varphi_{k}}{\lVert\varphi_{k}\rVert},
\]
for $k=1,2,3$, the corresponding coherence%
\[
\mu(g_{k},h_{k})\geq\lvert\langle g_{k},h_{k}\rangle\rvert\rightarrow1
\]
as $n\rightarrow\infty$. For any values of $\mu_{1},\mu_{2},\mu_{3}\in[0,1]$
such that \eqref{eq3} holds, i.e.$\ \mu_{1}\mu_{2}\mu_{3}<1/(r-1)=1$, we
cannot possibly have $\mu(g_{k},h_{k})\leq\mu_{k}$ for all $k=1,2,3$ since%
\[
\mu(g_{1},h_{1})\mu(g_{2},h_{2})\mu(g_{3},h_{3})\rightarrow1
\]
as $n\rightarrow\infty$.

\subsection{Uniqueness and minimality via coherence\label{sec:uniquecoh}}

In order to relate uniqueness and minimality of multilinear decompositions
to coherence, we need a simple observation about  the notion of
Kruskal rank introduced in Definition~\ref{kruskal-def}.

\begin{lemma}
\label{krank>coherence-lem}Let $\Phi\subseteq L^{2}(X_{1}\times\dots\times
X_{d})$ be finite and $\operatorname{krank} \Phi <\dim \operatorname{span} \Phi$. Then
\begin{equation}
\operatorname{krank}{\Phi}\geq\frac{1}{\mu(\Phi)}. \label{krank>coherence-eq}%
\end{equation}

\end{lemma}

\begin{proof}
Let $s=\operatorname{krank}{\Phi}+1$. Then there exists a subset of $s$
distinct unit vectors in $\Phi$, $\{\varphi_{1},\dots,\varphi_{s}\}$ such that
$\alpha_{1}\varphi_{1}+\dots+\alpha_{s}\varphi_{s}=0$ with $\lvert\alpha_{1}%
\rvert=\max\{\lvert\alpha_{1}\rvert,\dots,\lvert\alpha_{s}\rvert\}>0$. Taking
inner product with $\varphi_{1}$ we get $\alpha_{1}=-\alpha_{2}\langle\varphi
_{2},\varphi_{1}\rangle-\dots-\alpha_{s}\langle\varphi_{s},\varphi_{1}\rangle$ and so
$\lvert\alpha_{1}\rvert\leq(\lvert\alpha_{2}\rvert+\dots+\lvert\alpha
_{s}\rvert) \mu(\Phi)$.  Dividing by $\lvert\alpha_{1}\rvert$\ then yields
$1\leq(s-1)\mu(\Phi)$. The condition $\operatorname{krank}{\Phi}<\dim \operatorname{span} \Phi$ prevents $\Phi$ from being orthonormal, so $\mu(\Phi) > 0$ and we obtain \eqref{krank>coherence-eq}.
\end{proof}

We now characterize the uniqueness of the rank-retaining decomposition in
terms of coherence introduced in Definition~\ref{coherence-def}.

\begin{theorem}
\label{uniqueness-prop}Suppose $f\in L^{2}(X_{1}\times\dots\times X_{d})$ has a
multilinear decomposition%
\[
f=\sum_{p=1}^{r}\lambda_{p}\varphi_{1p}\otimes\dots\otimes\varphi_{dp}%
\]
where $\Phi_{k}:=\{\varphi_{k1},\dots,\varphi_{kr}\}$ are elements in
$L^{2}(X_{k})$ of unit norm and $\operatorname{krank} \Phi_k <\dim \operatorname{span} \Phi_k $ for all $k=1,\dots,d$. Let $\omega_{k}=\omega (\Phi_{k})$. If%
\begin{equation}
\sum_{k=1}^{d}\omega_{k}\geq2r-1, \label{Krus-omega-eq}%
\end{equation}
then $r=\operatorname{rank}(f)$ and the decomposition is essentially unique. In terms of coherence,
\eqref{Krus-omega-eq} takes the form
\begin{equation}
\sum_{k=1}^{d}\frac{1}{\mu_{k}} \ge 2r+d-1.\label{Krus-mu-eq}
\end{equation}
\end{theorem}
\begin{proof}
Inequality~\eqref{Krus-mu-eq} implies that $\sum_{k=1}^{d}\mu_{k}^{-1} \ge 2r+d-1$, where $\mu_{k}$ denotes $\mu(\Phi_{k})$. If it is satisfied, then so is Kruskal's condition \eqref{Kruskal-eq} thanks to Lemma~\ref{krank>coherence-lem}. The result hence directly follows from Lemma~\ref{Kruskal-lem} and Definition~\ref{def:relcohe}.
\end{proof}

Note that unlike the Kruskal ranks in \eqref{Kruskal-eq}, the coherences in
\eqref{Krus-mu-eq} are trivial to compute. In addition to uniqueness, an easy
but important consequence of Theorem~\ref{uniqueness-prop} is that it
provides a \textit{readily checkable sufficient condition} for tensor rank, which is
NP-hard over any field \cite{Haas90:ja, HillL09}.

Since the purpose of Theorem~\ref{uniqueness-prop} is to provide a computationally feasible alternative of Lemma~\ref{Kruskal-lem}, excluding the case $\operatorname{krank} \Phi_k =\dim \operatorname{span} \Phi_k $ is not an issue. Note that $\operatorname{krank} \Phi_k =\dim \operatorname{span} \Phi_k $ iff $\Phi_k$ comprises linearly independent elements, and the latter can be checked in polynomial time. So this is a case where Lemma~\ref{Kruskal-lem} can be readily checked and one does not need  Theorem~\ref{uniqueness-prop}.

\subsection{Existence and uniqueness via coherence\label{sec:bothcoh}}

The following existence and uniqueness sufficient condition may now be deduced
from Theorems~\ref{existence-prop} and \ref{uniqueness-prop}.

\begin{corollary}\label{final-cor}
If $d\geq3$ and if coherences $\mu_{k}$ satisfy
\begin{equation}
\left(  \prod_{k=1}^{d}\mu_{k}\right)  ^{1/d}\leq\frac{d}{2r+d-1}
\label{final-eq}%
\end{equation}
then the bounded coherence best rank-$r$ approximation problem has a unique
solution up to unimodulus scaling.
\end{corollary}

\begin{proof}
The existence in the case $r=1$ is assured, because the set of separable
functions $\{\varphi_{1}\otimes\dots\otimes\varphi_{d}:\varphi_{k}\in
L^{2}(X_{k})\}$ is closed. Consider thus the case $r\geq2$. Since the function
$f(x)=\frac{1}{x}-\left(  \frac{d}{2x+d-1}\right)  ^{d}$ is strictly positive
for $x\geq2$ and $d\geq3$, condition \eqref{final-eq} implies that
$\prod_{k=1}^{d}\mu_{k}$ is smaller than $1/r$, which permits to claim that
the solution exists by calling for Theorem~\ref{existence-prop}. Next in
order to prove uniqueness, we use the inequality between harmonic and
geometric means: if \eqref{final-eq} is verified, then we also necessarily
have $d\left(  \sum_{k=1}^{d}\mu_{k}^{-1}\right)  ^{-1}\leq\frac{d}{2r+d-1}$.
Hence $\sum_{k=1}^{d}\mu_{k}^{-1}\geq2r+d-1$ and we can apply
Theorem~\ref{uniqueness-prop}.
\end{proof}

In practice, simpler expressions than \eqref{final-eq} can be more attractive
for computational purposes. These can be derived from the inequalities between
means:%
\[
\left[  \frac{1}{d}\sum_{k=1}^{d}\mu_{k}^{-1}\right]^{-1}\leq\left[
\prod_{k=1}^{d}\mu_{k}\right]^{\frac{1}{d}}\leq\frac{1}{d}\sum_{k=1}^{d}\mu_{k}%
\leq\left[ \frac{1}{d}\sum_{k=1}^{d}\mu_{k}^{2}\right]^{\frac{1}{2}}.
\]
Examples of stronger sufficient conditions that could be used in place of
\eqref{final-eq} include
\begin{align}
\sum_{k=1}^{d}\mu_{k}  &  \leq\frac{d^{2}}{2r+d-1},\label{sufficient1-eq}\\
\sum_{k=1}^{d}\mu_{k}^{2}  &  \leq d\left(  \frac{d}{2r+d-1}\right)  ^{2}.
\label{sufficient2-eq}%
\end{align}

Another simplification can be performed, which yields differentiable
expressions of the constraints if \eqref{sufficient2-eq} is to be used. In
fact, noting that for any set of numbers $x_{1},\dots,x_{n}\in\mathbb{C}$,
$\max_{i=1,\dots,n}\lvert x_{i}\rvert\leq\sqrt{\sum_{i=1}^{n}\lvert
x_{i}\rvert^{2}}$, a sufficient condition ensuring that \eqref{sufficient2-eq}
is satisfied, and hence \eqref{final-eq}, is
\[
\sum_{k=1}^{d}\sum_{p<q}\lvert\langle\varphi_{kp},\varphi_{kq}\rangle
\rvert^{2}\leq d\left(  \frac{d}{2r+d-1}\right)  ^{2}.
\]

\subsection{Exact recoverability via coherence\label{sec:exact}}

We now describe a result that follows from the remarkable work of Temlyakov.
It allows us to in principle determine the multilinear decomposition meeting
the type of coherence conditions in Section~\ref{sec:exist}.

Some additional notations would be useful. We let $\Phi\subseteq\{f\in
L^{2}(X_{1}\times\dots\times X_{d}:\operatorname*{rank}(f)=1\}$ be a
\textit{dictionary\footnote{A \textit{dictionary} is any set $\Phi
\subseteq\mathbb{H}$ whose linear span is dense in the Hilbert space
$\mathbb{H}$.}} of separable functions (i.e.~rank-$1$) in $L^{2}(X_{1}%
\times\dots\times X_{d})$ that meets a bounded coherence condition , i.e.%
\begin{equation}
\mu(\Phi)<\mu\label{eq:dict}%
\end{equation}
for some $\mu\in[0,1)$ to be chosen later. Recall that the elements of $\Phi$
are implicitly assumed to be of unit norm (cf.~remark after
Definition~\ref{coherence-def}).

Let $t\in(0,1]$. The \textit{weakly orthogonal greedy algorithm}
(\textsc{woga}) is simple to describe: Set $f_{0}=f$. For each $m\in
\mathbb{N}$, we inductively define a sequence of $f_{m}$'s as follows:

\begin{enumerate}
\item $g_{m}\in\Phi$ is any element satisfying%
\[
\lvert\langle f_{m-1},g_{m}\rangle\rvert\geq t\sup_{g\in \Phi}
\lvert\langle f_{m-1},g\rangle\rvert;
\]

\item $h_{m}\in L^{2}(X_{1}\times\dots\times X_{d})$ is a projection of $f$
onto $\operatorname{span}(g_{1},\dots,g_{m})$, i.e.%
\begin{equation}
h_{m}\in\operatorname{argmin}\{\lVert f-g\rVert:g\in\operatorname{span}%
(g_{1},\dots,g_{m})\}; \label{proj}%
\end{equation}

\item $f_{m}\in L^{2}(X_{1}\times\dots\times X_{d})$ is a deflation of $f$ by
$h_{m}$, i.e.%
\[
f_{m}=f-h_{m}.
\]

\end{enumerate}

Note that deflation alone, without the coherence requirement, generally does not work for
computing  multilinear decompositions \cite{StegC10:laa}.
The following result, adapted here for our purpose, was proved for any arbitrary dictionary in
\cite{Temlyakov06}.

\begin{theorem}[Temlyakov]\label{thm:Tem}
Suppose $f\in L^{2}(X_{1}\times\dots\times X_{d})$ has a
multilinear decomposition%
\[
f=\sum_{p=1}^{r}\lambda_{p}\varphi_{1p}\otimes\dots\otimes\varphi_{dp}%
\]
with $\varphi_{1p}\otimes\dots\otimes\varphi_{dp}\in\Phi$ and the condition
that%
\[
r<\frac{t}{1+t}\left(  1+\frac{1}{\mu}\right)
\]
for some $t\in(0,1]$. Then the \textsc{woga}
algorithm recovers the factors exactly, or more precisely, $f_{r}=0$ and thus
$f=h_{r}$.
\end{theorem}

%Observe that we have chose .
So $h_{r}$, by its definition in \eqref{proj} and our choice of $\Phi$, is
given in the form of a linear combination of rank-$1$ functions, i.e., an
rank-$r$ multilinear decomposition.

\subsection{Greedy approximation bounds via coherence\label{sec:greedy}}

This discussion in Section~\ref{sec:exact} pertains to exact recovery of a
rank-$r$ multilinear decomposition although our main problem really takes the
form of a best rank-$r$ approximation more often than not. We will describe
some greedy approximation bounds for the approximation problem in this section.

We let%
\[
\sigma_{r}(\hat{f}):=\inf_{\boldsymbol{\alpha}\in\mathbb{C}^{r},\;\lVert
\varphi_{kp}\rVert=1}\left\Vert \hat{f}-\sum_{p=1}^{r}\alpha_{p}\prod
_{k=1}^{d}\varphi_{kp}\right\Vert .
\]
By our definition of rank and border rank,%
\begin{align*}
\sigma_{r}(\hat{f})  &  =\inf\{\lVert\hat{f}-f\rVert:\operatorname*{rank}%
(f)\leq r\}\\
&  =\min\{\lVert\hat{f}-f\rVert:\overline{\operatorname*{rank}}(f)\leq r\}.
\end{align*}
It would be wonderful if greedy algorithms along the lines of what we
discussed in Section~\ref{sec:exact} could yield an approximant within some
provable bounds that is a factor of $\sigma_{r}(\hat{f})$. However this is too
much to hope for mainly because a dictionary comprising all separable
functions, i.e., $\{f:\operatorname*{rank}(f)=1\}$ is far too large to be
amenable to such analysis. This does not prevent us from considering somewhat
more restrictive dictionaries like what we did in the previous section. So
again, let $\Phi\subseteq\{f\in L^{2}(X_{1}\times\dots\times X_{d}%
):\operatorname*{rank}(f)=1\}$ be such that%
\[
\mu(\Phi)<\mu
\]
for some given $\mu\in[0,1)$ to be chosen later. Let us instead define%
\[
s_{r}(\hat{f})=\inf_{\boldsymbol{\alpha}\in\mathbb{C}^{r},\;\varphi_{p}\in
\Phi}\left\Vert \hat{f}-\sum_{p=1}^{r}\alpha_{p}\varphi_{p}\right\Vert .
\]
Clearly%
\begin{equation}
\sigma_{r}(\hat{f})\leq s_{r}(\hat{f}) \label{eq4}%
\end{equation}
since the infimum is taken over a smaller dictionary.

The special case where $t=1$ in the \textsc{woga} described in
Section~\ref{sec:exact} is also called the \textit{orthogonal greedy
algorithm} (\textsc{oga}). The result we state next comes from the work of a
number of people done over the last decade: \eqref{eq:gms} is due to Gilbert,
Muthukrisnan, and Strauss in 2003 \cite{GMS}; \eqref{eq:tropp} is due to Tropp
in 2004 \cite{Tropp}; \eqref{eq:det} is due to Dohono, Elad, and Temlyakov in
2006 \cite{DET}; and \eqref{eq:liv} is due to Livshitz in 2012 \cite{Liv}. We
merely apply these results to our approximation problem here.

\begin{theorem}
Let $\hat{f}\in L^{2}(X_{1}\times\dots\times X_{d})$ and $f_{r}$ be the $r$th
iterate as defined in \textsc{woga} with $t=1$ and input $\hat{f}$.

\begin{enumerate}[\upshape (i)]
\item If $r<\frac{1}{32}\mu^{-1}$, then%
\begin{equation}
\lVert\hat{f}-f_{r}\rVert\leq8r^{1/2}s_{r}(\hat{f}). \label{eq:gms}%
\end{equation}

\item If $r<\frac{1}{3}\mu^{-1}$, then%
\begin{equation}
\lVert\hat{f}-f_{r}\rVert\leq(1+6r)^{1/2}s_{r}(\hat{f}). \label{eq:tropp}%
\end{equation}

\item If $r\leq\frac{1}{20}\mu^{-2/3}$, then%
\begin{equation}
\lVert\hat{f}-f_{r\log r}\rVert\leq24s_{r}(\hat{f}). \label{eq:det}%
\end{equation}

\item If $r\leq\frac{1}{20}\mu^{-1}$, then%
\begin{equation}
\lVert\hat{f}-f_{2r}\rVert\leq3s_{r}(\hat{f}). \label{eq:liv}%
\end{equation}

\end{enumerate}
\end{theorem}

It would be marvelous if one could instead establish bounds in \eqref{eq:gms},
\eqref{eq:tropp}, \eqref{eq:det}, and \eqref{eq:liv} with $\sigma_{r}(\hat
{f})$ in place of $s_{r}(\hat{f})$ and $\{f:\operatorname*{rank}(f)=1\}$ in
place of $\Phi$, dropping the coherence $\mu$ altogether. In which case one
may estimate how well the $r$th \textsc{oga} iterates $f_{r}$ approximates the
best rank-$r$ approximation. This appears to be beyond present capabilites.

We would to note that although the approximation theoretic technqiues (coherence,
greedy approximation, redundant dictionaries, etc) used in this article owe their newfound popularity
to compressive sensing, they owe their roots to works of the Russian school of approximation
theorists (e.g., Boris Kashin, Vladimir Temlyakov, et al.) dating back to the 1980s. We refer
readers to the  bibliography of \cite{Temlyakov} for more information.

\subsection{Checking coherence-based conditions}\label{sec:cc}

Since the conditions in Theorems~\ref{existence-prop}, \ref{uniqueness-prop}, \ref{thm:Tem}, and
Corollary~\ref{final-cor} all involve coherence, we will say a brief word about its computation.

It has recently been established that computing the spark of a finite set of vectors $\Phi$, i.e.,
the size of the smallest linearly dependent subset of $\Phi$, is strongly NP-hard \cite{TP}. Since
$\operatorname{spark} \Phi=\operatorname{krank}\Phi+1$, it immediately follows that the same is true for Kruskal rank.
\begin{corollary}[Kruskal rank is NP-hard]\label{cor:nphard}
Let $\mathbb{H}$ be a Hilbert space and $\Phi \subseteq \mathbb{H}$ be finite. The optimization problem of 
computing
\[
\operatorname{krank}(\Phi) := \max \left\{ k : \text{all } \Psi \in \binom{\Phi}{k}\text{ linearly independent}\right\}
\]
is strongly NP-hard.  $\binom{\Phi}{k} =$  set of all $k$-element subsets of $\Phi$.
\end{corollary}

Given the NP-hardness of Kruskal rank, one expects that 
Lemma~\ref{Kruskal-lem}, as well as its finite-dimensional counterparts \cite{Krus77:laa,SB},
would be computationally intractable to apply in reality. The reciprocal of coherence is therefore
a useful surrogate for Kruskal rank by virtue of Lemma~\ref{krank>coherence-lem} and the fact
that computing $\mu(\Phi)$ requires only $r(r-1)/2$ inner products $\langle \varphi_i, \varphi_j \rangle$, $i \ne j$,
where $r = \lvert \Phi \rvert$.

For finite-dimensional problems, i.e., $X_1,\dots,X_d$ are all finite sets, we may regard $\Phi$ as a matrix in $\mathbb{C}^{n \times r}$, and 
$\mu(\Phi)$ is simply the largest entry in magnitude in the Gram matrix $\Phi^{\mathsf{T}} \Phi$, which may be
rapidly computed using Strassen-type algorithms \cite{Strassen} in numerically stable ways \cite{DDH}.
For infinite-dimensional problems, the problem depends on the cost of integrating complex-valued functions defined
on $X_k$, $k = 1,\dots, d$. For example, if $X_1,\dots, X_d$ are all finite intervals of $\mathbb{R}$, one may use
\textit{inner product quadratures} \cite{Chen} to efficiently compute the Gram matrix
 $(\langle \varphi_i,\varphi_j\rangle)_{i,j=1}^n \in \mathbb{R}^{n \times n}$ and thereby find $\mu(\Phi)$.

\section{Applications\label{applications-sec}}

Various applications, many under the headings\footnote{ Other than
\textsc{candecomp} and \textsc{parafac}, the finite-dimensional multilinear decompositions have also been studied under the
names \textsc{cp}, \textsc{cand}, canonical decomposition, and canonical polyadic
decompositions.} of \textsc{candecomp}
\cite{CarrC70:psy} and \textsc{parafac} \cite{Hars70:ucla}, have appeared in
psychometrics and, more recently, also other data analytic applications. We
found that many of these applications suffer from a regretable defect --- there
are no compelling reasons nor rigorous arguments that support the use of a
rank-$r$ multilinear decomposition model. The mere fact that a data set may be
cast in the form of a $d$-dimensional array $A\in\mathbb{C}^{n_{1}\times
\dots\times n_{d}}$ does not mean that \eqref{CP-eq} would be the right or
even just a reasonable thing to do. In particular, how would one even
interpret the factors $\varphi_{kp}$'s when $d>2$? When $d=2$, one could
arguably interpret these as principal or varimax components when
orthonormality is imposed but for general $d>2$, a convincing application of a
model based on the rank-$r$ multilinear decomposition \eqref{CP-eq} must rely
on careful arguments that follow from first principles.

The goal of this section is two-fold. First we provide a
selection of applications where the rank-$r$ multilinear decomposition
\eqref{CP-eq} arises naturally via considerations of first principles (in electrodynamics, quantum mechanics, wave propagation, etc).  Secondly, we demonstrate that the coherence conditions discussed extensively in
Section~\ref{angular-sec} invariably have reasonable interpretations in terms
of physical quantities.

The use of a rank-$r$ multilinear decomposition model in signal processing via
higher-order statistics has a long history
\cite{PesqM01:ieeeit,FernFM08:SP,Como04:ijacsp,CardM95:ieeesp,ShamG93:SP}. Our
signal processing applications here are of a different nature, they are based
on geometrical properties of sensor arrays instead of considerations of
higher-order statistics. This line of argument first appeared in the work of
Sidiropoulos and Bro \cite{SidiBG00:ieeesp}, which is innovative and
well-motivated by first principles. However, like all other applications
considered thus far, whether in data analysis, signal processing,
psychometrics, or chemometrics, it does not address the serious nonexistence
problem that we discussed at length in Section~\ref{sec:exist}. Without any
guarantee that a solution to \eqref{fund1} exists, one can never be sure when
the model would yield a solution. Another issue of concern is that the Kruskal
uniqueness condition in Lemma~\ref{Kruskal-lem} has often been invoked to
provide evidence of a unique solution but we now know that this condition is practically impossible 
to check because of Corollary~\ref{cor:nphard}. The applications
considered below would use the coherence conditions developed in
Section~\ref{angular-sec} to avoid these difficulties. More precisely,
Theorem~\ref{existence-prop}, Theorem~\ref{uniqueness-prop}, and
Corollary~\ref{final-cor} are invoked to guarantee the existence of a solution
to the approximation problem and provide readily checkable conditions for
uniqueness of the solution, all via the notion of coherence.  Note that unlike Kruskal's
condition, which applies only to an \textit{exact} decomposition, Corollary~\ref{final-cor}
gives uniqueness of an approximation in noisy circumstances. 

In this section, applications are presented in finite dimension. In order to
avoid any confusion, $X^{\ast}$, $X^{\mathsf{H}}$ and $X^{\mathsf{T}}$ will
denote complex conjugate, hermitian transpose, and transpose, of the matrix
$X$ respectively.

\subsection{Joint channel and source estimation\label{JointChannelSource-sec}}

Consider a narrow band transmission problem in the far field. We assume here
that we are in the context of wireless telecommunications, but the same
principle could also apply in other areas. Let $r$ signals impinge on an array, so
that their mixture is recorded. We wish to recover the original signals
and to estimate their directions of arrival and respective powers at the
receiver. If the channel is specular, some of these signals can correspond to
different propagation paths of the same radiating source, and are therefore
correlated. In other words, $r$ does not denote the number of sources, but the
total number of distinct paths viewed from the receiver.

In the present framework, we assume that channels can be time-varying, but
that they can be regarded to be constant over a sufficiently short observation
length. The goal is to be able to work with extremely short samples.

In order to face this challenge, we assume that the sensor array is
structured, as in \cite{SidiBG00:ieeesp}. More precisely, the sensor array comprises
a \textit{reference array} with $n_{1}$ sensors, whose
location is defined by a vector $\mathbf{b}_{i}\in\mathbb{R}^{3}$, and
$n_{2}-1$ other subarrays obtained from the reference array by a translation
in space defined by a vector $\boldsymbol{\Delta}_{j}\in\mathbb{R}^{3}$,
$j =2,\dots, n_{2}$. The reference subarray is numbered with $j=1$ in the remainder.

Under these assumptions, the signal received at discrete time $t_{k}$,
$k=1,\dots,n_{3}$, on the $i$th sensor of the reference subarray can be
written as
\[
s_{i,1}(k)=\sum_{p=1}^{r}\sigma_{p}(t_{k})\exp(\psi_{i,p})
\]
with $\psi_{i,p}=\jmath\frac{\omega}{C}(\mathbf{b}_{i}^{\mathsf{T}}%
\mathbf{d}_{p})$ where the dotless $\jmath$ denotes $\sqrt{-1}$; 
$\mathbf{d}_{p} \in \mathbb{R}^3$ is of unit norm and denotes direction of arrival of the
$p$th path, $C$ denotes the wave celerity, and $\omega$ denotes the pulsation. Next, on the $j$th subarray, $j=2, \dots, n_2$, we have
\begin{equation}
s_{i,j}(k)=\sum_{p=1}^{r}\sigma_{p}(t_{k})\exp(\psi_{i,j,p})
\label{obsModel-eq}%
\end{equation}
with $\psi_{i,j,p}=\jmath\frac{\omega}{C}(\mathbf{b}_{i}^{\mathsf{T}%
}\mathbf{d}_{p}+\boldsymbol{\Delta}_{j}^{\mathsf{T}}\mathbf{d}_{p})$. If we
let $\boldsymbol{\Delta}_{1} =\mathbf{0}$, then \eqref{obsModel-eq}
also applies to the reference subarray. The crucial feature of this structure is
that variables $i$ and $j$ decouple in the function $\exp(\psi_{i,j,p})$, yielding
a relation resembling the rank-retaining multilinear decomposition:
\[
s_{i,j}(k)=\sum_{p=1}^{r}\lambda_{p}u_{ip}v_{jp}w_{kp}%
\]
where $u_{ip}=\exp\left(  \jmath\frac{\omega}{C}\mathbf{b}_{i}^{\top
}\mathbf{d}_{p}\right)  $, $v_{jp}=\exp\left(  \jmath\frac{\omega}%
{C}\boldsymbol{\Delta}_{j}^{\mathsf{T}}\mathbf{d}_{p}\right)  $ and
$w_{kp}=\sigma_{p}(t_{k})/\Vert\boldsymbol{\sigma}_{p}\Vert$, $\lambda
_{p}=\Vert\boldsymbol{\sigma}_{p}\Vert$.

By computing a rank-retaining decomposition of the hypermatrix
$S=(s_{i,j}(k))\in\mathbb{C}^{n_{1}\times n_{2}\times n_{3}}$, one may
jointly estimate: (i)~signal waveforms $\sigma_{p}(k)$, and (ii)~directions of arrival $\mathbf{d}_{p}$ of each propagation path, provided
$\mathbf{b}_{i}$ or $\boldsymbol{\Delta}_{j}$ are known.

However, the observation model \eqref{obsModel-eq} is not realistic, and an
additional error term should be added in order to account for modeling
inaccuracies and background noise. It is customary (and realistic thanks to
the central limit theorem) to assume that this additive error has a continuous
probability distribution, and that therefore the hypermatrix $S$ has the \textit{generic rank}. 
Since the generic rank is at least as large as $\lceil n_{1}n_{2}%
n_{3}/(n_{1}+n_{2}+n_{3}-2)\rceil$, which is always larger than Kruskal's
bound \cite{ComoLA09:jchemo}, we are led to the problem of
approximating the hypermatrix $S$ by another of rank $r$. We have seen that the
angular constraint imposed in Section~\ref{angular-sec} permits us to deal with a
well-posed problem. In order to see the physical meaning of this constraint,
we need to first define the tensor product between sensor subarrays.

\subsection{Tensor product between sensor subarrays\label{antennaProduct-sec}}

The sensor arrays we encounter are structured, in the sense that the whole
array is generated by one subarray defined by the collection of vector
locations $\{\mathbf{b}_{i}\in\mathbb{R}^{3}:1\leq i\leq n_{1}\}$, and a
collection of translations in space, $\{\boldsymbol{\Delta}_{j}\in
\mathbb{R}^{3}:1\leq j\leq n_{2}\}$. If we define vectors
\begin{align}
\mathbf{u}_{p}  &  =\frac{1}{\sqrt{n_{1}}}\left[  \exp\left(  \jmath
\frac{\omega}{C}\mathbf{b}_{i}^{\mathsf{T}}\mathbf{d}_{p}\right)  \right]
_{i=1}^{n_{1}},\nonumber\\
\mathbf{v}_{p}  &  =\frac{1}{\sqrt{n_{2}}}\left[  \exp\left(  \jmath
\frac{\omega}{C}\boldsymbol{\Delta}_{j}^{\mathsf{T}}\mathbf{d}_{p}\right)
\right]  _{j=1}^{n_{2}},\label{notationUp-eq}\\
\mathbf{w}_{p}  &  = \boldsymbol{\sigma}_{p}/\Vert\boldsymbol{\sigma
}_{p}\Vert,\nonumber
\end{align}
then we may view all measurements as the superimposition of
decomposable hypermatrices
$\lambda_{p}\mathbf{u}_{p}\otimes\mathbf{v}_{p}\otimes\mathbf{w}_{p}$.

Geometrical information of the sensor array is contained in $\mathbf{u}_{p}%
\otimes\mathbf{v}_{p}$ while energy and time information on each path $p$ is contained
in $\lambda_{p}$ and $\mathbf{w}_{p}$ respectively. Note that the
reference subarray and the set of translations play symmetric roles, in the
sense that $\mathbf{u}_{p}$ and $\mathbf{v}_{p}$ could be interchanged without
changing the whole array. This will become clear with a few examples.

When we are given a structured sensor array, there can be several ways of
splitting it into a tensor product of two (or more) subarrays, as shown in the following
simple examples.
\begin{example}
Define the matrix of sensor locations
\[
[\mathbf{b}_{1},\mathbf{b}_{2},\mathbf{b}_{3}]=
\begin{bmatrix}
0 & 0 & 1\\
0 & 1 & 1
\end{bmatrix}.
\]
This subarray is depicted in Figure~\ref{arrayV-fig}(b). By translating it
via the translation in Figure~\ref{arrayV-fig}(c) one obtains
another subarray. The union of the two subarrays yields the array of
Figure~\ref{arrayV-fig}(a). The same array is obtained by interchanging roles of
the two subarrays, i.e., three subarrays of two sensors deduced from each
other by two translations.
\end{example}

\begin{figure}[h]
\begin{center}
\includegraphics[scale=0.5]{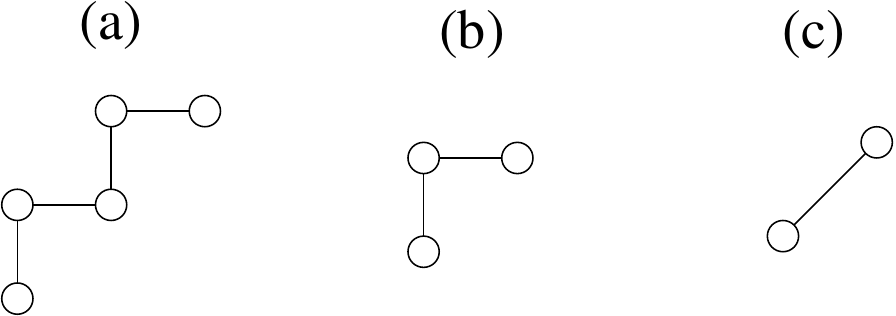}
\end{center}
\caption{Antenna array (a) obtained as tensor product of subarrays
(b) \& (c)}%
\label{arrayV-fig}%
\end{figure}

\begin{example}
Define the array by
\[
[\mathbf{b}_{1},\mathbf{b}_{2},\dots,\mathbf{b}_{6}]=
\begin{bmatrix}
0 & 1 & 2 & 0 & 1 & 2\\
0 & 0 & 0 & 1 & 1 & 1
\end{bmatrix}.
\]
This array, depicted in Figure~\ref{array1-fig}(a), can either be obtained from
the union of subarray of Figure~\ref{array1-fig}(b) and its translation defined
by Figure~\ref{array1-fig}(c), or from the array of Figure~\ref{array1-fig}(c)
translated three times according to Figure~\ref{array1-fig}(b). We 
express this relationship as
\[
\raisebox{-1.2ex}{\includegraphics[scale=0.4]{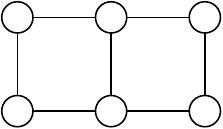}} = \raisebox{-1.2ex}{\includegraphics[scale=0.4]{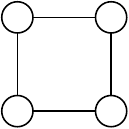}}\otimes
\raisebox{-1.2ex}{\includegraphics[scale=0.4]{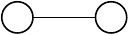}} = \raisebox{-1.2ex}{\includegraphics[scale=0.4]{subarray2h}}\otimes
\raisebox{-1.2ex}{\includegraphics[scale=0.4]{subarray4}}
\]
Another decomposition may be obtained as
\[
\raisebox{-1.2ex}{\includegraphics[scale=0.4]{subarray6}} = \raisebox{-1.2ex}{\includegraphics[scale=0.4]{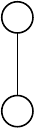}}\otimes
\raisebox{-1.2ex}{\includegraphics[scale=0.4]{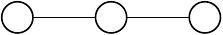}} = \raisebox{-1.2ex}{\includegraphics[scale=0.4]{subarray3}}\otimes
\raisebox{-1.2ex}{\includegraphics[scale=0.4]{subarray2v}}
\]
In fact,
$\raisebox{-1.2ex}{\includegraphics[scale=0.4]{subarray4}}=\raisebox{-1.2ex}{\includegraphics[scale=0.4]{subarray2v}}\otimes
\raisebox{-1.2ex}{\includegraphics[scale=0.4]{subarray2h}}$ and
$\raisebox{0ex}{\includegraphics[scale=0.4]{subarray3}}=\raisebox{0ex}{\includegraphics[scale=0.4]{subarray2h}}\otimes
\raisebox{0ex}{\includegraphics[scale=0.4]{subarray2h}}$. However, it is
important to stress that the various decompositions of the whole array into
tensor products of subarrays are not equivalent from the point of view of
performance. In particular, the Kruskal bound can be different, as we will see next.
\end{example}
\begin{figure}[tbh]
\begin{center}
\includegraphics[scale=0.5]{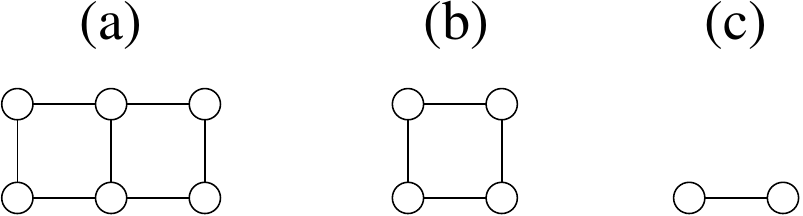}
\end{center}
\caption{Antenna array (a) obtained as tensor product of subarrays
(b) \& (c)}%
\label{array1-fig}%
\end{figure}

Similar observations can be made for grid arrays in general.
\begin{example}
Take an array of $9$ sensors located at $(x,y)\in\{1,2,3\}\times\{1,2,3\}$. We
have the relations
\[
\raisebox{-2.3ex}{\includegraphics[scale=0.4]{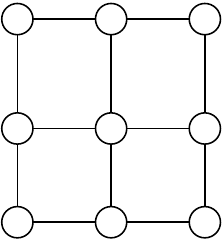}}=\raisebox{-1.2ex}{\includegraphics[scale=0.4]{subarray6}}\otimes
\raisebox{-1.2ex}{\includegraphics[scale=0.4]{subarray2v}}=\raisebox{-1.2ex}{\includegraphics[scale=0.4]{subarray4}}\otimes
\raisebox{-1.2ex}{\includegraphics[scale=0.4]{subarray4}}=\raisebox{0ex}{\includegraphics[scale=0.4]{subarray3}}\otimes
\raisebox{-2.3ex}{\includegraphics[scale=0.4]{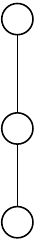}}
\]
among others.
\end{example}

Let us now have a look at the maximal number of sources $r_{\max}$ that can be
extracted from a $n_{1}\times n_{2}\times n_{3}$ hypermatrix in the absence of noise. A sufficient condition is
that the total number of paths, $r$, is smaller than Kruskal's bound
\eqref{Kruskal-eq}. We shall simplify the bound by making two assumptions:
(a)~the loading matrices are generic, i.e., they are of full rank, and (b)~the
number of paths is larger than the sizes $n_{1}$ and $n_{2}$ of the two
subarrays entering the array tensor product, and smaller than the number of
time samples, $n_{3}$. Under these simplifying assumptions, Kruskal's bound
becomes $2r_{\max}\leq n_{1}+n_{2}+r_{\max}-2$, or:
\begin{equation}
r_{\max}=n_{1}+n_{2}-2 \label{simpleBound-eq}%
\end{equation}
The table below illustrates the fact that the choice of subarrays has an
impact on this bound.

\begin{center}%
\begin{tabular}
[c]{|c|c|cc|c|}\hline
Array & Subarray & $n_{1}$ & $n_{2}$ & $r_{\max}$\\
~ & product & ~ & ~ & ~\\\hline\hline
\raisebox{-1ex}[4ex][2ex]{\includegraphics[scale=0.3]{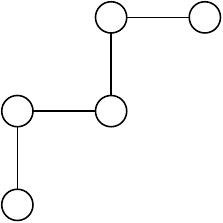}} &
\raisebox{-1ex}{\includegraphics[scale=0.3]{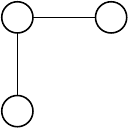}} $\otimes$
\raisebox{-1ex}{\includegraphics[scale=0.3]{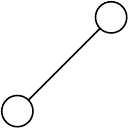}} & $3$ & $2$ &
$3$\\\hline
\raisebox{-1ex}[3.5ex][2ex]{\includegraphics[scale=0.3]{subarray6}} &
\raisebox{-1ex}{\includegraphics[scale=0.3]{subarray4}} $\otimes$
\raisebox{0ex}{\includegraphics[scale=0.3]{subarray2h}} & $4$ & $2$ &
$4$\\\cline{2-5}%
~ & \raisebox{-0.8ex}{\includegraphics[scale=0.3]{subarray2v}} $\otimes$
\raisebox{0ex}[2.5ex][2ex]{\includegraphics[scale=0.3]{subarray3}} & $2$ & $3$
& $3$\\\hline
~ & \raisebox{-2ex}[4ex][3ex]{\includegraphics[scale=0.3]{subarray3v}}
$\otimes$ \raisebox{0ex}{\includegraphics[scale=0.3]{subarray3}} & $3$ & $3$ &
$4$\\\cline{2-5}%
\raisebox{-1ex}[3ex][2ex]{\includegraphics[scale=0.3]{subarray9}} &
\raisebox{-1ex}{\includegraphics[scale=0.3]{subarray6}} $\otimes$
\raisebox{-1ex}{\includegraphics[scale=0.3]{subarray2v}} & $6$ & $2$ &
$6$\\\cline{2-5}%
~ & \raisebox{-1ex}[3ex][2ex]{\includegraphics[scale=0.3]{subarray4}}
$\otimes$ \raisebox{-1ex}{\includegraphics[scale=0.3]{subarray4}} & $4$ & $4$
& $6$\\\hline
\end{tabular}
\end{center}

\subsection{Significance of the angular constraint\label{signification-sec}}

We are now in a position to interpret the meanings of the various coherences in light of this application. According to the notations given in
\eqref{notationUp-eq}, the first coherence
\[
\mu_{1}=\max_{p\neq q}\, \lvert\mathbf{u}_{p}^{\mathsf{H}}\mathbf{u}_{q}\rvert
\]
corresponds to the angular separation viewed from the reference subarray. The
vectors $\mathbf{b}_{i}$ and $\mathbf{d}_{p}$ have unit norms, as
do the vectors $\mathbf{u}_{p}$. The quantity $\lvert\mathbf{u}%
_{p}^{\mathsf{H}}\mathbf{u}_{q}\rvert$ may thus be viewed as a measure of angular
separation between $\mathbf{d}_{p}$ and $\mathbf{d}_{q}$, as we shall demonstrate 
in Proposition~\ref{collinear-prop}.

\begin{definition}
\label{resolvent-def}We shall say that a collection of vectors $\{\mathbf{b}%
_{1},\dots,\mathbf{b}_{n}\}$ is \textbf{resolvent} with respect to a direction
$\mathbf{v}\in\mathbb{R}^{3}$ if there exist two indices $k$ and $l$ such that
$\mathbf{v}=\mathbf{b}_{k}-\mathbf{b}_{l}$ and%
\[
0<\Vert\mathbf{v}\Vert<\frac{\lambda}{2},
\]
where $\lambda= 2\pi C/\omega$ denotes the wavelength.
\end{definition}

Let $\mathbf{b}_{i}$, $\mathbf{d}_{p}$ and $\mathbf{u}_{q}$ be defined as in
\eqref{notationUp-eq}, $i = 1,\dots, n_{1}$, $p, q =1,\dots, n_{2}$. Then we
have the following.
\begin{proposition}
\label{collinear-prop} If $\{\mathbf{b}_{1},\dots,\mathbf{b}_{n}\}$ is
resolvent with respect to three linearly independent directions, then
\[
\lvert\mathbf{u}_{p}^{\mathsf{H}}\mathbf{u}_{q}\rvert=1 \quad \Leftrightarrow \quad
\mathbf{d}_{p}=\mathbf{d}_{q}.
\]

\end{proposition}

\begin{proof}
Assume that $\lvert\mathbf{u}_{p}^{\mathsf{H}}\mathbf{u}_{q}\rvert=1$. Since
they are of unit norm, vectors $\mathbf{u}_{p}$ and $\mathbf{u}_{q}$ are
collinear with a unit modulus proportionality factor. Hence, from
\eqref{notationUp-eq}, for all $j,k = 1,\dots, n_{1}$, $(\mathbf{b}%
_{j}-\mathbf{b}_{k})^{\mathsf{T}}(\mathbf{d}_{p}-\mathbf{d}_{q})\in
\lambda\mathbb{Z}$, where $\lambda$ is as in
Definition~\ref{resolvent-def}. Since $\{\mathbf{b}_{1},\dots,\mathbf{b}%
_{n}\}$ is resolvent, there exist $k$, $l$ such that $0<\Vert\mathbf{b}%
_{k}-\mathbf{b}_{l}\Vert<\lambda/2$. As the vectors $\mathbf{d}_{p}$
are of unit norm, $\Vert\mathbf{d}_{p}-\mathbf{d}_{q}\Vert\leq2$ and we
necessarily have that $(\mathbf{b}_{k}-\mathbf{b}_{l})^{\mathsf{T}}%
(\mathbf{d}_{p}-\mathbf{d}_{q})=0$, i.e., $\mathbf{d}_{p}-\mathbf{d}_{q}$
is orthogonal to $\mathbf{b}_{k}-\mathbf{b}_{l}$. The same
reasoning can be carried out with the other two independent vectors. The
vector $\mathbf{d}_{p}-\mathbf{d}_{q}$ must be $\mathbf{0}$ because it is orthogonal to
three linearly independent vectors in $\mathbb{R}^{3}$. The converse is
immediate from the definition of $\mathbf{u}_{q}$.
\end{proof}

Note that the condition in Definition~\ref{resolvent-def} is not very
restrictive, since sensor arrays usually contain sensors separated by half a
wavelength or less. Thanks to Proposition~\ref{collinear-prop}, we now know
that uniqueness of the matrix factor $U=[\mathbf{u}_{1}%
,\dots,\mathbf{u}_{r}]$ and the identifiability of the directions of arrival
$\mathbf{d}_{p}$ are equivalent. By the results of
Section~\ref{angular-sec}, the uniqueness can be ensured by a constraint on coherence such
as \eqref{final-eq}.

As in Section~\ref{antennaProduct-sec}, the second coherence may be interpreted as a measure of the
minimal angular separation between paths, viewed from the subarray defining translations.

The third coherence is the maximal correlation coefficient
between signals received from various paths on the array
\[
\mu_{3}=\max_{p\neq q}\frac{\lvert\boldsymbol{\sigma}_{p}^{\mathsf{H}%
}\boldsymbol{\sigma}_{q}\rvert}{\Vert\boldsymbol{\sigma}_{p}\Vert
\Vert\boldsymbol{\sigma}_{q}\Vert}.
\]

In conclusion, the best rank-$r$ approximation exists and is unique if either
signals propagating through various paths are not too correlated, or if
their direction of arrival are not too close, where ``not too'' is taken to
mean that the product of coherences satisfies inequality
\eqref{final-eq} of Corollary~\ref{final-cor}. In other words, one can
separate paths with arbitrarily high correlation provided they are sufficiently well
separated in space.

Hence, the decomposition of a sensor array into a tensor product of two (or more)
sensor subarrays depends not only on Kruskal's bound, as elaborated in
Section~\ref{antennaProduct-sec}, but also on the ability of the latter subarrays
to separate two distinct directions of arrival (cf.\ Proposition~\ref{collinear-prop}).

\subsection{CDMA communications}

The application to antenna array processing we described in
Section~\ref{JointChannelSource-sec} also applies to all source separation problems
\cite{ComoJ10}, provided an additional diversity is available. An example is the case
of Code Division Multiple Access (CDMA)
communications. In fact, as pointed out in \cite{SidiGB00:ieeesp}, it
is possible to distinguish between symbol and chip diversities. We will elaborate on the latter example.

Consider a downlink CDMA communication with $r$ users, each assigned a
spreading sequence $C_{p}(k)$, $p=1,\dots, r$, $k=1,\dots, n$. Denote by
$A_{ip}$ the complex gain between sensor $i$, $i=1,\dots,m$, and user $p$,
by $S_{jp}$ the symbol sequence transmitted by user $p$, $j\in\mathbb{Z}$, and
by $H_{p}(k)$ the channel impulse response of user $p$. The signal received
on sensor $i$ during the $k$th chip of the $j$th symbol period takes the
form
\[
T_{ijk}=\sum_{p=1}^{r}A_{ip}S_{jp}B_{kp}%
\]
where $B_{kp}=\sum_{t}H_{p}(k-t)C_{p}(t)$ denotes the output of the $p$th
channel excited by the $p$th coding sequence, upon removal of the guard chips
(which may be affected by two different consecutive symbols) \cite{SidiGB00:ieeesp}.

The columns of matrix $B = (B_{kp}) \in \mathbb{R}^{n \times r}$ are often referred to as ``effective
codes'', and coincide with spreading codes if the channel is memoryless and
noiseless. In practice, the receiver filter is matched to the transmitter
shaping filter combined with the propagation channel, so that effective and
spreading codes are ideally proportional. Under these conditions, the
coherence $\mu_{C}$ accounts for the angular separation between spreading
sequences: $\mu_{C}=0$ means that they are orthogonal. On the other hand,
$\mu_{A}=0$ means that the symbol sequences are all uncorrelated. Lastly, as seen
in Proposition~\ref{collinear-prop}, $\mu_{B}=1$ means that the directions of arrival
are collinear.

In order to avoid multiple access interferences, spreading sequences are
usually chosen to be uncorrelated for all delays, which implies  that
they are orthogonal. However, the results obtained in Section~\ref{angular-sec} show that
spreading sequences \textit{do not need to be orthogonal}, and symbol sequences \textit{need not
be uncorrelated}, as long as the directions of arrival are not collinear. In particular,
shorter spreading sequences may be used for the same number of users, which
increases throughput. Alternatively, for a given spreading gain, one may increase the number of
users. These are possible because the coherence conditions
in Section~\ref{angular-sec} allow one to relax the constraint
of having almost orthogonal spreading sequences. On the other hand, some
directions of arrival may be collinear if the corresponding spreading
sequences are sufficiently well separated angularly. These conclusions are
essentially valid when users are synchronized, i.e., for downlink communications.

\subsection{Polarization}

The use of polarization as an additional diversity has its roots in \cite{NehoP94:TSP}.
Several attempts to use this diversity in the framework of tensor-based source localization and estimation
can be found in the literature \cite{GuoMBZL11:TSP}.

In this framework, we consider again an array of $n_{1}$ sensors, whose
locations are given by $\mathbf{b}_{i} \in \mathbb{R}^{3}$, $i =1,\dots,n$. We
assume a narrow-band transmission in the far field (i.e., sources, or source
paths, are all seen as plane waves at the receiver sensor array). The
difference with Section \ref{antennaProduct-sec} is that translation
diversity is not mandatory anymore, provided that the impinging waves are polarized and that their polarization is neither linear nor circular. One measures
the electric and magnetic fields at each sensor as a function of time, so that
$n_{2}=6$. More precisely, $\mathbf{v}_{p}$ of \eqref{notationUp-eq} is replaced by
\begin{equation}
\mathbf{v}_{p}=B_{p}\mathbf{g}_{p}%
\end{equation}
where $B_{p} \in \mathbb{R}^{6\times 2}$ depends only on the direction of
arrival $\mathbf{d}_{p}$ (defined in Section~\ref{antennaProduct-sec}), and
$\mathbf{g}_{p} \in \mathbb{C}^{2}$ depends only on the orientation and ellipticity of the
polarization of the $p$th wave.

Coherences $\mu_{1}$ and $\mu_{3}$ are the same as in Section~\ref{signification-sec}, and represent respectively the angular
separation between directions of arrival, and correlation between
arriving sources. It is slightly more difficult to see the significance of $\mu_{2}$,
the coherence associated with polarization.

For this, we need to go into more details \cite{NehoP94:TSP}. Let  $\alpha_{p}\in(-\pi/2,\pi/2]$
and $\beta_{p}\in(-\pi/4,0)\cup(0,\pi/4)$ denote respectively the orientation
and ellipticity angles of the polarization of the $p$th wave. Let $\theta_{p}\in[0,2\pi)$ and $\phi_{p}\in(-\pi/2,\pi/2]$ denote
respectively the azimuth and elevation of the direction of arrival of the $p$th path. We have
\[
B_{p}=\frac{1}{\sqrt{2}}%
\begin{bmatrix}
\mathbf{e}_{p} & \mathbf{f}_{p}\\
\mathbf{f}_{p} & -\mathbf{e}_{p}%
\end{bmatrix}
, \qquad\mathbf{g}_{p}=Q(\alpha_{p})\mathbf{h}_{p},
\]
where
\begin{gather*}
\mathbf{e}_{p}=%
\begin{bmatrix}
-\sin\theta_{p}\\
~\cos\theta_{p}\\
0
\end{bmatrix}
,\qquad\mathbf{f}_{p}=%
\begin{bmatrix}
-\cos\theta_{p}\sin\phi_{p}\\
-\sin\theta_{p}\sin\phi_{p}\\
\cos\phi_{p}%
\end{bmatrix},
\\
Q(\alpha)=%
\begin{bmatrix}
\cos\alpha & \sin\alpha\\
-\sin\alpha & \cos\alpha
\end{bmatrix}
,\qquad\mathbf{h}_{p}=%
\begin{bmatrix}
\cos\beta_{p}\\
\jmath\sin\beta_{p}%
\end{bmatrix}.
\end{gather*}
The unit vector defining the $p$th direction of arrival is
\[
\mathbf{d}_{p}=%
\begin{bmatrix}
\cos\theta_{p}\cos\phi_{p}\\
\sin\theta_{p}\cos\phi_{p}\\
\sin\phi_{p}%
\end{bmatrix}.
\]
So the triplet $(\mathbf{d}_{p},\mathbf{e}_{p},\mathbf{f}_{p})$ forms a
right orthonormal triad.

\begin{lemma}
\label{polar-lemma}$\lvert\mathbf{g}_{p}^{\mathsf{H}}\mathbf{g}_{q}\rvert=1$
if and only if $\alpha_{p}=\alpha_{q}+k\pi$ and $\beta_{p}=\beta_{q}$,
$k\in\mathbb{Z}$.
\end{lemma}
\begin{proof}
First note that $Q(\alpha_{p})^{\mathsf{H}%
}Q(\alpha_{q})=Q(\alpha_{q}-\alpha_{p})$. Hence $\mathbf{g}_{p}^{\mathsf{H}%
}\mathbf{g}_{q}$ can be of unit modulus only if $\mathbf{h}_{p}$ and
$Q(\alpha_{q}-\alpha_{p})\mathbf{h}_{q}$ are collinear. But the first entry of
$\mathbf{h}_{p}$ is real and the second is purely imaginary. So the
corresponding imaginary and real parts of $Q(\alpha_{q}-\alpha_{p}%
)\mathbf{h}_{q}$ must be zero, which implies that $\sin(\alpha_{q}-\alpha
_{p})=0$. Consequently $Q(\alpha_{q}-\alpha_{p})=\pm I$, which yields
$\mathbf{h}_{p}=\pm\mathbf{h}_{q}$. But because the angle $\beta$ lies in the
interval $(-\pi/4,\pi/4)$, only the positive sign is acceptable.
\end{proof}

\begin{proposition}
\label{polar-prop}$\lvert\mathbf{v}_{p}^{\mathsf{H}}\mathbf{v}_{q}\rvert\leq
1$, with equality if and only if $\alpha_{p}=\alpha_{q}+k\pi$, $\beta
_{p}=\beta_{q}$, $\theta_{p}=\theta_{q}+k^{\prime}\pi$ and $\phi_{p}=\phi_{q}%
$, $k,k^{\prime}\in\mathbb{Z}$.
\end{proposition}
\begin{proof}
We have $\lvert\mathbf{v}%
_{p}^{\mathsf{H}}\mathbf{v}_{q}\rvert=\lvert\mathbf{g}_{p}^{\mathsf{H}}%
B_{p}^{\mathsf{T}}B_{q}\mathbf{g}_{q}\rvert$. Notice that the matrix
$B_{p}^{\mathsf{T}}B_{q}$ is of the form
\[
B_{p}^{\mathsf{T}}B_{q}=%
\begin{bmatrix}
\gamma & \eta\\
-\eta & \gamma
\end{bmatrix}
\]
where $\gamma$ and $\eta$ are real, $\gamma=\frac{1}{2}(\mathbf{e}%
_{p}^{\mathsf{T}}\mathbf{e}_{q}+\mathbf{f}_{p}^{\mathsf{T}}\mathbf{f}_{q})$
and $\eta=\frac{1}{2}(\mathbf{e}_{p}^{\mathsf{T}}\mathbf{f}_{q}-\mathbf{f}%
_{p}^{\mathsf{T}}\mathbf{e}_{q})$. Since $\mathbf{g}_{p}$ and
$\mathbf{g}_{q}$ are of unit norms, $\lvert\mathbf{v}_{p}^{\mathsf{H}%
}\mathbf{v}_{q}\rvert$ can be of unit modulus only if $B_{p}%
^{\mathsf{T}}B_{q}$ has an eigenvalue of unit modulus, which requires that
$\gamma^{2}+\eta^{2}=1$. We now prove that $\gamma^{2}+\eta^{2}\leq1$
with equality if and only if the four sets of equalities hold.

With this goal in mind, define the $6$-dimensional vectors
\[
\mathbf{z}=\frac{1}{\sqrt{2}}%
\begin{bmatrix}
\mathbf{e}_{p}\\
\mathbf{f}_{p}%
\end{bmatrix}
,\quad\mathbf{w}=\frac{1}{\sqrt{2}}%
\begin{bmatrix}
\mathbf{e}_{q}\\
\mathbf{f}_{q}%
\end{bmatrix}
,\quad\mathbf{w}^{\prime}=\frac{1}{\sqrt{2}}%
\begin{bmatrix}
\mathbf{f}_{q}\\
-\mathbf{e}_{q}%
\end{bmatrix}
.
\]
Then $\gamma=\mathbf{z}^{\mathsf{T}}\mathbf{w}$ and $\gamma=\mathbf{z}%
^{\mathsf{T}}\mathbf{w}^{\prime}$. Decompose $\mathbf{z}$ into two
orthogonal parts: $\mathbf{z}=\mathbf{z}_{0}+\mathbf{z}_{1}$, with
$\mathbf{z}_{0}\in\operatorname{span}\{\mathbf{w},\mathbf{w}^{\prime}\}$ and
$\mathbf{z}_{0}\bot\mathbf{z}_{1}$. Clearly, $\gamma^{2}+\eta^{2}%
=\Vert\mathbf{z}_{0}\Vert^{2}$. Moreover,
$\Vert\mathbf{z}_{0}\Vert^{2}\leq\Vert\mathbf{z}\Vert^{2}=1$, with equality if
and only if $\mathbf{z}\in\operatorname{span}\{\mathbf{w},\mathbf{w}^{\prime
}\}$. By inspection of the definitions of $\mathbf{e}_{p}$ and $\mathbf{e}%
_{q}$, we see that the third entry of $\mathbf{z}$ and $\mathbf{w}$ is $\mathbf{0}$.
Hence $\mathbf{z}\in\operatorname{span}\{\mathbf{w},\mathbf{w}^{\prime}\}$ is
possible only if either $\mathbf{z}$ is collinear to $\mathbf{w}$ or if the
third entry of $\mathbf{w}^{\prime}$ is $\mathbf{0}$. In the latter case, it means
that $\phi_{q}=\pi/2$, and so $\phi_{p}=\pi/2$ and $\theta_{p}%
=\theta_{q}$. In the former case, it can be seen that $\sin\theta_{p}%
=\sin\theta_{q}$, and finally that $\phi_{p}=\phi_{q}$.

The last step is to rewrite $\gamma$ and $\eta$ as a function of angle
$\theta_{p}-\theta_{q}$, using trigonometric relations: $\gamma=\cos
(\theta_{p}-\theta_{q})(1+\sin\phi_{p}\sin\phi_{q})+\cos\phi_{p}\cos\phi_{q}$
and $\eta=\sin(\theta_{p}-\theta_{q})(\sin\phi_{p}+\sin\phi_{q})$. This
eventually shows that $\gamma=1$ and $\eta=0$. As a consequence,
$\lvert\mathbf{v}_{p}^{\mathsf{H}}\mathbf{v}_{q}\rvert=1$ only if
$B_{p}^{\mathsf{T}}B_{q}=I$, and the result follows from Lemma~\ref{polar-lemma}.
\end{proof}

Proposition~\ref{polar-prop} shows that a
constraint on the coherence $\mu_{2}$ compels source paths to have either
different directions of arrival or different polarizations, giving $\mu_2$ physical meaning.

\subsection{Fluorescence spectral analysis\label{sec:fluo}}

Here is a well-known application to fluorescence spectral analysis originally discussed in \cite{SmilBG04}.
We use the notations in Example~\ref{eg:hitch}.
Suppose we have $l$ samples with an unknown number of pure substances in
different concentrations that are fluorescent. If $a_{ijk}$ is the
measured fluorescence emission intensity at wavelength $\lambda_{j}%
^{\operatorname{em}}$ of the $i$th sample excited with light of wavelength $\lambda
_{k}^{\operatorname{ex}}$. The measured data is a $3$-dimensional hypermatrix
$A=(a_{ijk})\in\mathbb{R}^{l\times m\times n}$. At low concentrations, Beer's
law of spectroscopy (which is in turn a consequence of fundamental principles
in quantum mechanics) can be linearized \cite{LuciMRB09:cils}, yielding a
rank-retaining decomposition
\[
A=\mathbf{x}_{1}\otimes\mathbf{y}_{1}\otimes\mathbf{z}_{1}+\dots
+\mathbf{x}_{r}\otimes\mathbf{y}_{r}\otimes\mathbf{z}_{r}.
\]
This reveals the true chemical factors responsible for the data:
$r=\operatorname{rank}(A)$ gives the number of pure substances in the
mixtures, $\mathbf{x}_{p}=(x_{1p},\dots,x_{lp})$ gives the relative
concentrations of $p$th substance in specimens $1,\dots,l$; $\mathbf{y}%
_{p}=(y_{1p},\dots,y_{mp})$ gives the excitation spectrum of $p$th substance; and
$\mathbf{z}_{p}=(z_{1p},\dots,z_{np})$ gives the emission spectrum of $p$th
substance. The emission and excitation spectra would then allow one to identify the pure substances.

Of course, this is only valid in an idealized situation when measurements
are performed perfectly without error and noise. Under realistic noisy
circumstances, one would then need to a find best rank-$r$ approximation, which is where the
coherence results of Section~\ref{angular-sec} play a role. In this case, $\mu (\mathbf{x}_1,\dots,\mathbf{x}_r)$
measures the relative abundance of the pure substances in the samples while $\mu (\mathbf{y}_1,\dots,\mathbf{y}_r)$ and $\mu (\mathbf{z}_1,\dots,\mathbf{z}_r)$
measure the spectroscopic likeness of these pure substances in the sense of absorbance and fluorescence respectively. 

\subsection{Statistical independence induces diversity}

We will now discuss a
somewhat different way to achieve diversity. Assume the linear model below
\begin{equation}
\mathbf{x}(t)=U\mathbf{s}(t), \label{linearModel-eq}%
\end{equation}
where only the signal $\mathbf{x}(t)$ is observed, $U=[\mathbf{u}_{1}%
,\dots,\mathbf{u}_{r}]$ is an unknown $n\times r$ mixing matrix, and $\mathbf{s}(t) = (s_1(t),\dots, s_r(t))$
has mutually statistically independent components. One may construct $K_d(\mathbf{x})$, the $d$th order cumulant hypermatrix \cite{Mccu87} of
$\mathbf{x}(t)$,  and it will satisfy the multilinear model
\[
K_d(\mathbf{x})= \sum_{p=1}^{r}\lambda_{p}(\mathbf{s}) \mathbf{u}_{p}\otimes\dots\otimes\mathbf{u}_{p}%
\]
where $\lambda_{p}(\mathbf{s}) $ denotes the $p$th diagonal entry of the $d$th cumulant hypermatrix of $\mathbf{s}$. Because of the
statistical independence of  $\mathbf{s}(t)$,  the off-diagonal entries of the $d$th cumulant hypermatrix of $\mathbf{s}$ are zero \cite{ComoJ10,Mccu87}.
If $d\geq3$, then the matrix $U$ and the entries $\lambda_{p}(\mathbf{s}) $ can be identified \cite{ComoJ10}. One may apply the results of Section~\ref{angular-sec} to 
deduce uniqueness of the solution.

Such problems generalize to convolutive mixtures and have applications in
telecommunications, radar, sonar, speech processing, and biomedical
engineering \cite{ComoJ10}.

\subsection{Nonstationarity induces diversity}

If a signal $x(t)$ is nonstationary, its time-frequency transform, defined by
\[
X(t,f)=\int x(u)\kappa(u-t;f)\,du
\]
for some given kernel $\kappa$, bears information. If variables $t$ and $f$
are discretized, then the values of $X(t,f)$ can be stored in a matrix $X$;
and the more nonstationary the signal $x(t)$, the larger the rank of $X$. A similar
statement can be made on a signal $y(\mathbf{z})$ depending on a spatial
variable $\mathbf{z}$. The discrete values of the space-wavevector transform
$Y(\mathbf{z},\mathbf{w})$ of a field $y(\mathbf{z})$ can be stored in a
matrix $Y$; and the less homogeneous the field $y(\mathbf{z})$, the larger the
rank of $Y$. This is probably the reason why algorithms proposed in
\cite{BeckCAHM12:SP, WeisRHJH09:taipei} permit one to localize and extract dipole
contributions in the brains using a multilinear model, provided that one
has distinct time-frequency or space-wavevector patterns. Nevertheless, such
localization is guaranteed to be successful only under restrictive assumptions.

\section{Further work}

A separate article discussing practical algorithms for the bounded coherence
best rank-$r$ multilinear approximation is under preparation with additional
coauthors. These algorithms follow the general strategy of the greedy
approximations \textsc{woga} and \textsc{oga} discussed in
Sections~\ref{sec:exact} and \ref{sec:greedy} but contain other elements
exploiting the special separable structure of our problem. Extensive numerical
experiments will be provided in the forthcoming article.

\section*{Acknowledgement}

We thank Ignat Domanov for pointing out that in
Theorem~\ref{existence-prop}, the factor on the right hand side may be
improved from $r$ to $r-1$, and that the improvement is sharp. We owe special thanks to Harm Derksen
for pointing out an error in an earlier version of Section~\ref{sec:nuclear} and for very helpful
discussions regarding nuclear norm of tensors.
Sections~\ref{sec:exact} and \ref{sec:greedy} came from an enlightening series
of lectures Vladimir Temlyakov gave at the IMA in Minneapolis and the useful
pointers he graciously provided afterwards. We gratefully acknowledge Tom
Luo, Nikos Sidiropoulos, Yuan Yao, and two anonymous reviewers for their helpful comments.

The work of LHL is partially supported by AFOSR Young Investigator Award FA9550-13-1-0133,
NSF Collaborative Research Grant DMS 1209136, and NSF CAREER Award DMS 1057064. The work of PC is funded
by the European Research Council under the European Community's Seventh
Framework Programme FP7/2007--2013 Grant Agreement no.~320594.

\end{document}